\documentclass[11pt,a4paper]{article}
\usepackage[T1]{fontenc}
\usepackage{lmodern}
\usepackage[margin=25mm]{geometry}
\usepackage{mathtools,amssymb,amsthm}
\usepackage{graphicx,array,booktabs,makecell}
\usepackage{enumitem,setspace,needspace}
\usepackage[super,nospace]{cite}

\usepackage[hidelinks]{hyperref}
\newcommand{\logit}{\operatorname{logit}}
\newcommand{\expit}{\operatorname{expit}}

\newcommand{\ESS}{\mathrm{ESS}}
\newcommand{\lnDOR}{\ln\mathrm{DOR}}
\newcommand{\SAUC}{\mathrm{SAUC}}
\newcommand{\bH}{\beta_{\mathrm{H}}}
\numberwithin{equation}{section}
\theoremstyle{plain}
\newtheorem{theorem}{Theorem}[section]
\newtheorem{proposition}[theorem]{Proposition}
\newtheorem{corollary}[theorem]{Corollary}
\AddToHook{env/proposition/before}{\Needspace{6\baselineskip}}
\AddToHook{env/corollary/before}{\Needspace{5\baselineskip}}
\makeatletter
\renewcommand{\@biblabel}[1]{#1.}
\makeatother

\newcommand{\RepositoryCommit}{82d8ca2}
\newcommand{\RepositoryVersion}{1.0.2}
\newcommand{\RepositorySnapshot}{https://github.com/Kazuki-Nishida/dta-small-study-effects/tree/82d8ca2df0795f419cc44f2ff4aa0e50bd956fcf}
\hypersetup{
pdftitle={Small-study effects on the hierarchical summary ROC curve: latent accuracy and threshold trends in meta-analysis of diagnostic test accuracy},
pdfauthor={Yuta Nonomiya and Kazuki Nishida},
pdfkeywords={diagnostic test accuracy, funnel plot, hierarchical summary ROC, publication bias, small-study effects, threshold effect}
}

\newtheorem*{thmdecomp}{Proposition 2.1 of the main text (HSROC decomposition of the size trend)}
\newtheorem*{corsize}{Corollary 2.2 of the main text (threshold-induced $\lnDOR$ trend), formal version}
\newtheorem*{thmcurve}{Proposition 2.3 of the main text (accuracy null and invariance of the summary curve)}

\begin{document}
\singlespacing
\begin{center}
{\Large\bfseries Small-study effects on the hierarchical summary ROC curve: latent accuracy and threshold trends in meta-analysis of diagnostic test accuracy\par}
\vspace{1.2em}
{\large Yuta Nonomiya$^{1,2}$ and Kazuki Nishida$^{3,4,*}$\par}
\vspace{0.9em}
{\small
$^1$Graduate School of Informatics, Osaka Metropolitan University, 1-1 Gakuen-cho, Naka-ku, Sakai, Osaka 599-8531, Japan\\[0.4em]
$^2$German Research Center for Artificial Intelligence (DFKI), Trippstadter Str. 122, 67663 Kaiserslautern, Germany\\[0.4em]
$^3$Department of Biostatistics, Nagoya University Graduate School of Medicine, 65 Tsurumai-cho, Showa-ku, Nagoya 466-8550, Japan\\[0.4em]
$^4$Department of Clinical Biostatistics, School of Public Health, Graduate School of Medicine, Kyoto University, Kyoto, Japan
\par}
\end{center}
\vspace{0.6em}
\noindent\textbf{Correspondence:} Kazuki Nishida, Department of Biostatistics, Nagoya University Graduate School of Medicine, 65 Tsurumai-cho, Showa-ku, Nagoya 466-8550, Japan. Email: \href{mailto:kazuki.nishida1988@gmail.com}{kazuki.nishida1988@gmail.com}.

\medskip\noindent\textbf{Author email:} Yuta Nonomiya, \href{mailto:nonomiya.yuta@omu.ac.jp}{nonomiya.yuta@omu.ac.jp}.

\medskip\noindent\textbf{ORCID:} Yuta Nonomiya, \href{https://orcid.org/0000-0003-0511-0244}{0000-0003-0511-0244}; Kazuki Nishida, \href{https://orcid.org/0000-0003-0367-8557}{0000-0003-0367-8557}.

\section*{Abstract}
Diagnostic meta-analyses commonly assess small-study effects using the Deeks test, which tests for a study-size trend in the log diagnostic odds ratio. Under the hierarchical summary receiver operating characteristic (HSROC) model, this quantity reflects both latent accuracy and threshold unless the curve is symmetric. A threshold trend can therefore generate a log diagnostic odds ratio trend without a latent accuracy trend, or conceal an existing one. We develop a likelihood-ratio test of the latent accuracy trend under a common-shape binomial HSROC model. Its null hypothesis states that the summary curve is unchanged across study sizes. A log diagnostic odds ratio test using the same binomial fit provides a comparison with a different null hypothesis. In simulations, the proposed test remained near its nominal level under its null and had higher power than Deeks with a latent accuracy trend alone. In a colorectal cancer review of faecal immunochemical tests, it detected a negative latent accuracy trend obscured on the log diagnostic odds ratio scale by an opposing threshold contribution. All three tests detected size trends for impedance plethysmography in deep-vein thrombosis, where the fitted curve was nearly symmetric. The proposed test distinguishes changes in the summary curve from changes in operating position.

\medskip\noindent\textbf{Keywords:} Diagnostic test accuracy; Funnel plot; Hierarchical summary ROC; Publication bias; Small-study effects; Threshold effect.

\clearpage
\onehalfspacing
\section{Introduction}\label{sec:intro}

Clinicians need to know whether a diagnostic test is accurate enough to guide patient care. This judgement often requires evidence from several studies. Sensitivity and specificity describe test performance at a given threshold. They are estimated from each study's $2\times2$ table of test results against a reference standard. To summarize these measures across study populations and settings while accounting for between-study heterogeneity, meta-analyses of diagnostic test accuracy (DTA) commonly use hierarchical models\cite{reitsma2005,cochraneDTA2023}.

Sensitivity and specificity also need to be interpreted together. Changing the threshold can increase one while decreasing the other, moving the operating point along a receiver operating characteristic (ROC) curve. Differences between studies may therefore reflect differences in the curve itself or in operating position along it. The hierarchical summary ROC (HSROC) model represents this distinction through latent accuracy and latent threshold\cite{rutter2001}. For a given curve shape, latent accuracy determines the position of the summary curve, whereas latent threshold determines the operating position along it.

A further concern in DTA meta-analysis is \emph{small-study effects}: systematic differences in reported results between smaller and larger studies\cite{egger1997}. The Deeks test is commonly used to assess small-study effects in DTA meta-analysis. It regresses each study's log diagnostic odds ratio ($\lnDOR$) on the inverse square root of its effective sample size. The corresponding Deeks funnel plot displays $\lnDOR$ against this measure of study size\cite{deeks2005,cochraneDTA2023}.

The diagnostic odds ratio, however, can depend on threshold\cite{buerkner2014}. Under the HSROC model, $\lnDOR$ reflects latent accuracy alone only when the summary curve is symmetric. Otherwise, it also contains a latent threshold component. As we show below, differences in latent threshold between smaller and larger studies can produce an association between $\lnDOR$ and study size even when latent accuracy is unchanged. Threshold differences can also mask small-study effects in latent accuracy when the two contributions oppose each other. A test based on $\lnDOR$ therefore does not, in general, answer whether the summary ROC curve differs across study sizes.

We therefore assess small-study effects in HSROC latent accuracy by testing the null hypothesis of no trend with study size. Under a common-shape model, this null states that the same summary ROC curve applies across study sizes. We derive the relationship between the latent accuracy null and the $\lnDOR$ null. We use the binomial likelihood of the HSROC model to construct a likelihood-ratio test. We compare the proposed test with the conventional Deeks test and with an $\lnDOR$ test using the same binomial fit.

This paper is organized as follows. Section~\ref{sec:methods} presents the model, hypotheses and testing procedures. Section~\ref{sec:sim} compares these three tests in simulations, and Section~\ref{sec:apps} applies them to two diagnostic reviews. Section~\ref{sec:disc} discusses the implications and limitations of the proposed approach.

\section{Model, hypotheses and test}\label{sec:methods}

\subsection{The HSROC model with study size}\label{sec:model}

Study $i$ ($i=1,\dots,k$) contributes a $2\times2$ table: $\mathrm{TP}_i$ of $n_{1i}$ diseased participants test positive, and $\mathrm{FP}_i$ of $n_{0i}$ non-diseased participants test positive. Let $\eta_i=\logit(\text{sensitivity}_i)$ and $\varphi_i=\logit(\text{false-positive rate}_i)$, so that
\begin{equation}\label{eq:binom}
\mathrm{TP}_i\sim\mathrm{Bin}\{n_{1i},\expit(\eta_i)\},\qquad \mathrm{FP}_i\sim\mathrm{Bin}\{n_{0i},\expit(\varphi_i)\}.
\end{equation}
Study size is represented by the effective sample size of Deeks et al.\cite{deeks2005}, $\ESS_i=4n_{1i}n_{0i}/(n_{1i}+n_{0i})$. We set $s_i=1/\sqrt{\ESS_i}$ and $x_i=s_i-\bar s$. A study-size trend denotes a slope with respect to $s_i$, or equivalently $x_i$; larger values of $s_i$ correspond to smaller studies.

We allow such trends in both latent components of the HSROC model\cite{rutter2001}:
\begin{equation}\label{eq:hsroc}
\eta_i=e^{-\bH/2}\Bigl(\theta_i+\frac{\alpha_i}{2}\Bigr),\quad
\varphi_i=e^{\bH/2}\Bigl(\theta_i-\frac{\alpha_i}{2}\Bigr),\quad
\theta_i\sim N(\Theta+\gamma_\theta x_i,\sigma_\theta^2)\perp\alpha_i\sim N(\Lambda+\gamma_\alpha x_i,\sigma_\alpha^2),
\end{equation}
where $\theta_i$ is latent threshold, $\alpha_i$ is latent accuracy, and $\gamma_\theta$ and $\gamma_\alpha$ are their study-size trends. The latent residuals are independent conditional on study size. Their variances and the shape parameter $\bH$ are constant across study sizes. Write
\begin{equation}\label{eq:lambda}
\lambda=e^{-\bH},
\end{equation}
so that at latent accuracy $\alpha$ the ROC curve is $\eta=\lambda\varphi+\lambda^{1/2}\alpha$ on the logit scale. It is symmetric if and only if $\lambda=1$.

To assess small-study effects in latent accuracy, we test the \emph{accuracy null}
\begin{equation}\label{eq:null}
H_0^{\alpha}\colon\ \gamma_\alpha=0,\qquad \gamma_\theta\ \text{unrestricted}.
\end{equation}
Accuracy and threshold refer throughout to the components of the common-shape model \eqref{eq:hsroc}.

The same model has a bivariate parametrization\cite{reitsma2005,chu2006}. Conditional on study size, $(\eta_i,\varphi_i)$ is bivariate normal with means $\mu_\eta+\beta_\eta x_i$ and $\mu_\varphi+\beta_\varphi x_i$, residual standard deviations $\sigma_\eta,\sigma_\varphi$ and residual correlation $\rho$. Under the correspondence of Harbord et al.\cite{harbord2007}, $\lambda=\sigma_\eta/\sigma_\varphi$, and the two parametrizations yield the same likelihood. We call $(\beta_\eta,\beta_\varphi)$ the \emph{trend vector}.

We compare the $\lnDOR$ coordinate $\eta_i-\varphi_i$ with the accuracy coordinate $\eta_i-\lambda\varphi_i$, which eliminates latent threshold. To express both, write $a_i(c)=\eta_i-c\varphi_i$, $c\ge0$, with study-size trend $\beta(c)=\beta_\eta-c\beta_\varphi$. The choices $c=1$ and $c=\lambda$ give the two coordinates. Geometrically, $a_i(c)$ is constant along the direction $(1,c)$ in the $(\varphi,\eta)$ plane.

\subsection{What each procedure tests}\label{sec:map}

The two contrasts give the null hypotheses $\beta_\eta-\beta_\varphi=0$ for $\lnDOR$ and $\beta_\eta-\lambda\beta_\varphi=0$ for latent accuracy. The following decomposition shows how latent threshold contributes to each contrast.

\begin{proposition}[HSROC decomposition of the size trend]\label{thm:decomp}
Under \eqref{eq:hsroc}, with $\lambda$ as in \eqref{eq:lambda}:
\begin{enumerate}[label=(\roman*)]
\item \emph{(Directions and exact elimination.)} A change of $\theta_i$ moves a study along $(1,\lambda)$ in the $(\varphi,\eta)$ plane and a change of $\alpha_i$ along $(-1,\lambda)$. For every study,
\begin{equation}\label{eq:elim}
\eta_i-\lambda\varphi_i=\lambda^{1/2}\alpha_i,\qquad \eta_i+\lambda\varphi_i=2\lambda^{1/2}\theta_i:
\end{equation}
the contrast $\eta_i-\lambda\varphi_i$ depends only on latent accuracy, whereas $\eta_i+\lambda\varphi_i$ depends only on latent threshold.
\item \emph{(Trends.)} $\beta_\eta=\lambda^{1/2}(\gamma_\theta+\gamma_\alpha/2)$ and $\beta_\varphi=\lambda^{-1/2}(\gamma_\theta-\gamma_\alpha/2)$; conversely $\gamma_\alpha=\lambda^{-1/2}(\beta_\eta-\lambda\beta_\varphi)$ and $\gamma_\theta=\tfrac12(\lambda^{-1/2}\beta_\eta+\lambda^{1/2}\beta_\varphi)$. The accuracy null \eqref{eq:null} is $\beta_\eta=\lambda\beta_\varphi$.
\item \emph{(Mixing.)} For every $c\ge0$,
\begin{equation}\label{eq:mix}
\beta(c)=\lambda^{-1/2}\Bigl[(\lambda-c)\,\gamma_\theta+\frac{\lambda+c}{2}\,\gamma_\alpha\Bigr],
\end{equation}
which is free of $\gamma_\theta$ for every value of $\gamma_\theta$ iff $c=\lambda$, where $\beta(\lambda)=\lambda^{1/2}\gamma_\alpha$.
\item \emph{(The symmetric case.)} The $\lnDOR$ axis carries $\gamma_\theta$ with weight $\lambda^{-1/2}(\lambda-1)$, which vanishes iff $\bH=0$. When $\bH=0$, $\alpha_i=\eta_i-\varphi_i=\lnDOR_i$ and $\theta_i=(\eta_i+\varphi_i)/2$ exactly, so the true $\lnDOR$ is the latent accuracy and its size trend is $\gamma_\alpha$.
\end{enumerate}
\end{proposition}

Proofs are in Section~A of the supplementary material. Under model \eqref{eq:hsroc}, the coordinate $\eta_i-\lambda\varphi_i$ eliminates the latent threshold for each study. The $\lnDOR$ coordinate makes the same elimination only when $\lambda=1$.

\begin{corollary}[Threshold-induced $\lnDOR$ trend]\label{cor:size}
Under a pure threshold trend, $\gamma_\alpha=0$ and $\gamma_\theta\neq0$ with $\bH\neq0$, the true $\lnDOR$ has the nonzero size trend $\beta(1)=\lambda^{-1/2}(\lambda-1)\gamma_\theta$.
\end{corollary}

In this setting, the $\lnDOR$ null is false even though the accuracy null holds. A two-sided $t$-type test based on a consistent estimator of $\beta(1)$ with estimated standard error $O_p(k^{-1/2})$ rejects the $\lnDOR$ null with probability tending to one as $k$ grows (Section~A of the supplementary material). We evaluate the corresponding rejection behaviour of the conventional Deeks test in Section~\ref{sec:sim}.

The sign of the induced trend is $\operatorname{sign}\{(\lambda-1)\gamma_\theta\}$: with $\lambda<1$ and smaller studies at more liberal operating positions ($\gamma_\theta>0$), small studies display \emph{lower} $\lnDOR$ despite unchanged latent accuracy. Such a reverse pattern can arise from a latent threshold trend without selective publication.

Table~\ref{tab:map} compares the null hypotheses and test procedures. Under model \eqref{eq:hsroc}, the $\lnDOR$ null corresponding to the Deeks assessment is $\beta_\eta-\beta_\varphi=0$. The accuracy null leaves the latent threshold trend unrestricted. The two null hypotheses coincide when $\lambda=1$; otherwise, a threshold trend can make them differ.

We also test the $\lnDOR$ null using a Wald contrast from the binomial fit employed by the proposed likelihood-ratio test. This comparison retains the $\lnDOR$ null hypothesis while replacing the empirical-logit regression with the binomial likelihood. Comparing the two null hypotheses within that fit then shows what changes when the assessment concerns latent accuracy.

\begin{table}[htb]
\centering
\caption{Null hypotheses and implementation of the three tests under model \eqref{eq:hsroc}. The shape is $\lambda=\exp(-\bH)$, and the logit and latent trend coefficients are related by Proposition~\ref{thm:decomp}. The accuracy null leaves the latent threshold trend unrestricted. Estimation and reference distributions are specified in Section~\ref{sec:inference}.}
\label{tab:map}
\footnotesize
\setlength{\tabcolsep}{6pt}
\begin{tabular}{@{}>{\raggedright\arraybackslash}p{90pt} >{\raggedright\arraybackslash}p{145pt} >{\raggedright\arraybackslash}p{155pt}@{}}
\toprule
Procedure & Null hypothesis in model \eqref{eq:hsroc} & Estimation and test \\
\midrule
Deeks test & $\beta_\eta-\beta_\varphi=0$ & Empirical-logit weighted regression; $t$ reference \\
\addlinespace[4pt]
Binomial-fit $\lnDOR$ test & $\beta_\eta-\beta_\varphi=0$ & Binomial maximum likelihood; Wald contrast \\
\addlinespace[4pt]
Proposed test & $\beta_\eta-\lambda\beta_\varphi=0$, equivalently $\gamma_\alpha=0$ & Binomial maximum likelihood; likelihood-ratio test \\
\bottomrule
\end{tabular}
\end{table}

\subsection{What the accuracy null means for the summary curve}\label{sec:curve}

Write $\alpha(x)=\Lambda+\gamma_\alpha x$ for the mean latent accuracy at size $x$ and
\begin{equation}\label{eq:curve}
S_x(u)=\expit\bigl\{\lambda\logit(u)+\lambda^{1/2}\alpha(x)\bigr\},\qquad u\in(0,1),
\end{equation}
for the summary curve evaluated at that mean latent accuracy, with sensitivity expressed as a function of the false-positive rate $u$.

\begin{proposition}[Accuracy null and invariance of the summary curve]\label{thm:curve}
Under \eqref{eq:hsroc}, with $\bH$ not depending on $x$, for every $u\in(0,1)$
\begin{equation}\label{eq:shift}
\frac{\partial S_x(u)}{\partial x}=\lambda^{1/2}\gamma_\alpha\,S_x(u)\{1-S_x(u)\}.
\end{equation}
Consequently the accuracy null \eqref{eq:null} holds iff the summary curve is the same for studies of all sizes, and under the alternative the sign of $\gamma_\alpha$ is the sign of the size trend of the sensitivity at every false-positive rate, and of the area under the summary curve.
\end{proposition}

The accuracy null therefore asks whether one summary curve applies across study sizes. A threshold trend changes the mean operating position without changing that curve; an accuracy trend shifts the curve at every false-positive rate. The proof is in Section~A of the supplementary material.

\subsection{The test}\label{sec:inference}

We fit the seven-parameter bivariate binomial meta-regression with study size on both logits, using maximum likelihood and adaptive Gauss--Hermite quadrature\cite{chu2006} (Section~B of the supplementary material). The binomial likelihood requires no continuity correction. The fitted shape is $\hat\lambda=\hat\sigma_\eta/\hat\sigma_\varphi$, and the latent trends follow from Proposition~\ref{thm:decomp}(ii). The equivalent HSROC parametrization gives the same likelihood\cite{harbord2007}.

The accuracy null imposes $\beta_\eta=\lambda\beta_\varphi$. We maximize the constrained likelihood with $\beta_\eta=\lambda^{1/2}\gamma_\theta$ and $\beta_\varphi=\lambda^{-1/2}\gamma_\theta$, re-estimating the shape and the remaining parameters. Let $\widehat\ell$ and $\widehat\ell_0$ denote the full and constrained maximum log-likelihoods. The likelihood-ratio statistic and the two-sided $p$ value used here are
\[
D=2(\widehat\ell-\widehat\ell_0),\qquad
p=2\bigl\{1-F_{t_{k-2}}(\sqrt D)\bigr\},
\]
where $F_{t_{k-2}}$ is the distribution function of a $t$ distribution with $k-2$ degrees of freedom. The asymptotic reference for $D$ is $\chi^2_1$; the $t_{k-2}$ reference for its signed square root is a finite-sample convention, evaluated against the $\chi^2_1$ reference in Section~C.2 of the supplementary material. The threshold null is tested analogously.

Standard errors of the latent trends use the delta method and the full covariance matrix, including shape uncertainty and its covariance with the trend estimates; intervals use $t_{k-2}$. The binomial-fit $\lnDOR$ test uses the Wald contrast $\hat\beta_\eta-\hat\beta_\varphi$ with the same $t$ reference. An auxiliary Wald test of the accuracy coordinate is reported in Section~C.2 of the supplementary material to examine the role of the test statistic.

For the Deeks test and funnel plots, we calculate empirical logits after adding $0.5$ to all four cells of every study's table\cite{kulinskaya2015}. The conventional Deeks test regresses the resulting empirical $\lnDOR$ on $s_i$, with weights $\ESS_i$\cite{deeks2005}.

Alongside the Deeks funnel, we plot the empirical \emph{accuracy coordinate} $\hat\eta_i-\hat\lambda(\hat\varphi_i-\hat\mu_\varphi)$, where $\hat\eta_i$ and $\hat\varphi_i$ are the study's empirical logits and $\hat\mu_\varphi$ is the fitted mean logit false-positive rate. This moves each observed point along the fitted threshold direction to the common false-positive rate $\expit(\hat\mu_\varphi)$. The overlaid model trend is $\hat\beta_\eta-\hat\lambda\hat\beta_\varphi=\hat\lambda^{1/2}\hat\gamma_\alpha$ per unit of $s_i$. The plotted points retain sampling variation. At $\hat\lambda=1$, the two plotted coordinates coincide up to a constant. More generally, the plotted coordinate is an affine transformation of the threshold-adjusted $\lnDOR$ in the Moses--Littenberg summary ROC formulation\cite{moses1993}. Writing $D_{\mathrm{ML}}=\eta-\varphi$ and $S_{\mathrm{ML}}=\eta+\varphi$ gives $\eta-\lambda\varphi=\{(\lambda+1)/2\}(D_{\mathrm{ML}}-bS_{\mathrm{ML}})$, where $b=(\lambda-1)/(\lambda+1)$ is the slope of the ROC line at fixed latent accuracy in the $(S_{\mathrm{ML}},D_{\mathrm{ML}})$ plane under \eqref{eq:hsroc}.

\section{Simulation study}\label{sec:sim}

\subsection{Design}\label{sec:simdesign}

We used simulations to assess rejection under the accuracy null and detection of latent accuracy trends, with and without a concurrent latent threshold trend. We considered four scenarios: neither trend, a threshold trend alone, an accuracy trend alone, and both trends.

Data were generated from model \eqref{eq:hsroc}, with baseline $k=30$, residual correlation $\rho=0.4$ and shapes $\lambda=1/4$, $1/2$, $1/\sqrt2$, $1$, $\sqrt2$, $2$ and $4$. The mean operating point was fixed at logit sensitivity $1$ and logit false-positive rate $-2$. We set $\sigma_\eta\sigma_\varphi=0.75^2$ and prevalence to $0.35$. Total study sizes were log-normal with median $300$ and clipped to $[40,4000]$. The generator is specified in Section~C.1 of the supplementary material.

All four scenarios were evaluated at each shape. The nonzero accuracy trend was $0.5$ per standard deviation of $s_i$, with smaller studies more accurate. The nonzero threshold trend was $0.4$ residual standard deviations per standard deviation of $s_i$, with smaller studies at more liberal positions. A further comparison fixed $\lambda=1/2$ and varied the threshold-trend strength from $0$ to $1.6$ in steps of $0.2$, with no accuracy trend.

Additional comparisons varied $k=10$, $20$, $30$ and $50$ at $\rho=0.4$, and $\rho=0$, $0.4$ and $0.8$ at $k=30$. These used $\lambda=1/2$, $1$ and $2$ under neither trend, a threshold trend alone and an accuracy trend alone. The latent accuracy change was held fixed across $\rho$. The design contains 35 baseline and 45 additional settings, with 1000 replicates per setting and an independent seed for each setting.

Each replicate was analysed by the Deeks test, the binomial-fit $\lnDOR$ test and the proposed latent-accuracy likelihood-ratio (LR) test, as defined in Section~\ref{sec:inference}. All tests were evaluated at the nominal level $0.10$. Rates were calculated from valid $p$ values; optimizer flags were retained, and computational diagnostics are reported for every setting in Section~C.4 of the supplementary material.

We report rejection rates for the respective null hypotheses. For the $\lnDOR$ tests, rejection when $\gamma_\alpha=0$ need not be a type I error: a pure threshold trend makes their own null false when $\lambda\neq1$. The accuracy-trend scenarios show how detection depends on the presence of the threshold trend.

\subsection{Results}\label{sec:simresults}

\begin{figure}[p]
\centering
\includegraphics[width=\textwidth]{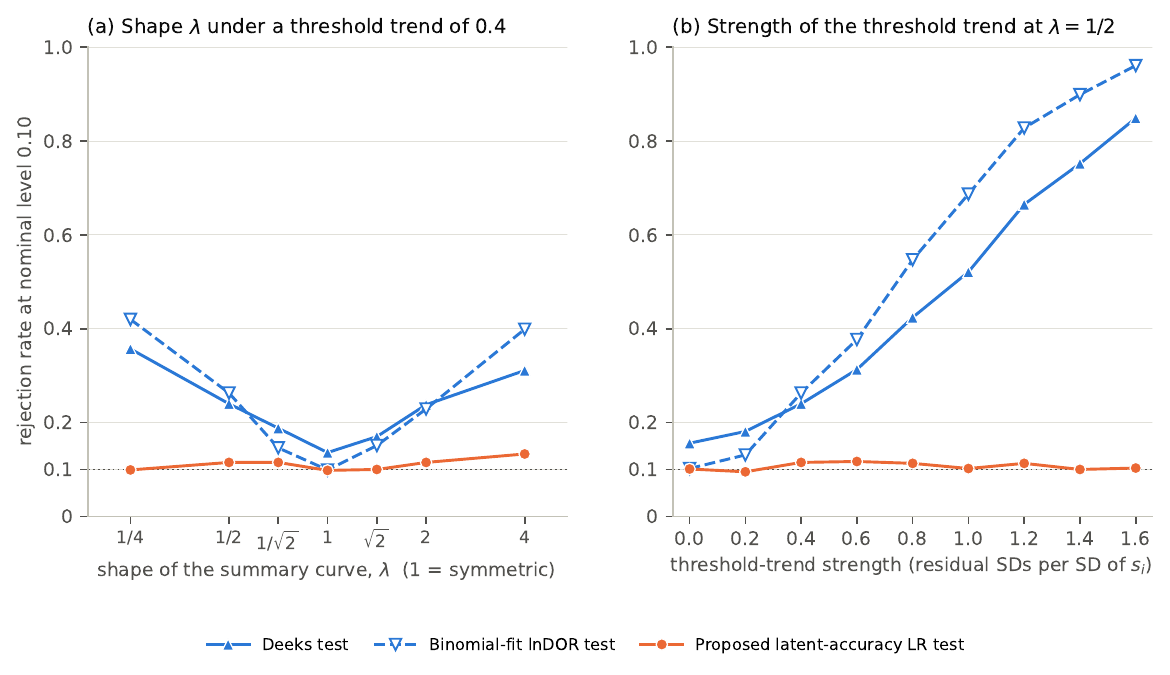}
\caption{Rejection under the accuracy null ($\gamma_\alpha=0$), with $k=30$ and residual correlation $\rho=0.4$. (a)~Shape $\lambda$ varies under a threshold trend of $0.4$ residual standard deviations per standard deviation of $s_i$. (b)~Threshold-trend strength varies at $\lambda=1/2$. Each point uses 1000 replicates; the dotted line marks the nominal level $0.10$. Blue: Deeks test (solid) and the binomial-fit $\lnDOR$ test (dashed); orange: proposed test. Generating values and rates are in Tables~S1 and S2 of the supplementary material.}
\label{fig:sim}
\par\vspace{0.3em}\begin{flushleft}\footnotesize Alt text: Two panels show rejection rates when latent accuracy has no size trend. In the left panel, the two blue ln DOR curves are lowest near shape one and rise toward either end. In the right panel, both blue curves rise as the threshold trend strengthens. The orange accuracy-test curve stays close to the nominal level in both panels.\end{flushleft}
\end{figure}

With neither trend present, the two binomial procedures had rejection rates closer to the nominal level than the Deeks test. The Deeks test rejected in $0.137$ to $0.157$ of the replicates, compared with $0.098$ to $0.124$ for the binomial-fit $\lnDOR$ test and $0.101$ to $0.121$ for the proposed test (Table~S1 of the supplementary material).

A threshold trend alone increased rejection by the two $\lnDOR$ tests on asymmetric curves despite unchanged latent accuracy (Figure~\ref{fig:sim}(a)). At $\lambda=1/4$ and $4$, Deeks rejected in $0.357$ and $0.311$ of the replicates; the binomial-fit $\lnDOR$ test gave rates of $0.420$ and $0.399$. The proposed test ranged from $0.098$ to $0.133$ across the seven shapes.

The contrast strengthened as the threshold trend increased at $\lambda=1/2$ (Figure~\ref{fig:sim}(b)). At a strength of $1.6$ residual standard deviations, rejection reached $0.849$ for Deeks and $0.961$ for the binomial-fit $\lnDOR$ test. The proposed test remained between $0.095$ and $0.117$ over the sweep.

\begin{figure}[p]
\centering
\includegraphics[width=\textwidth]{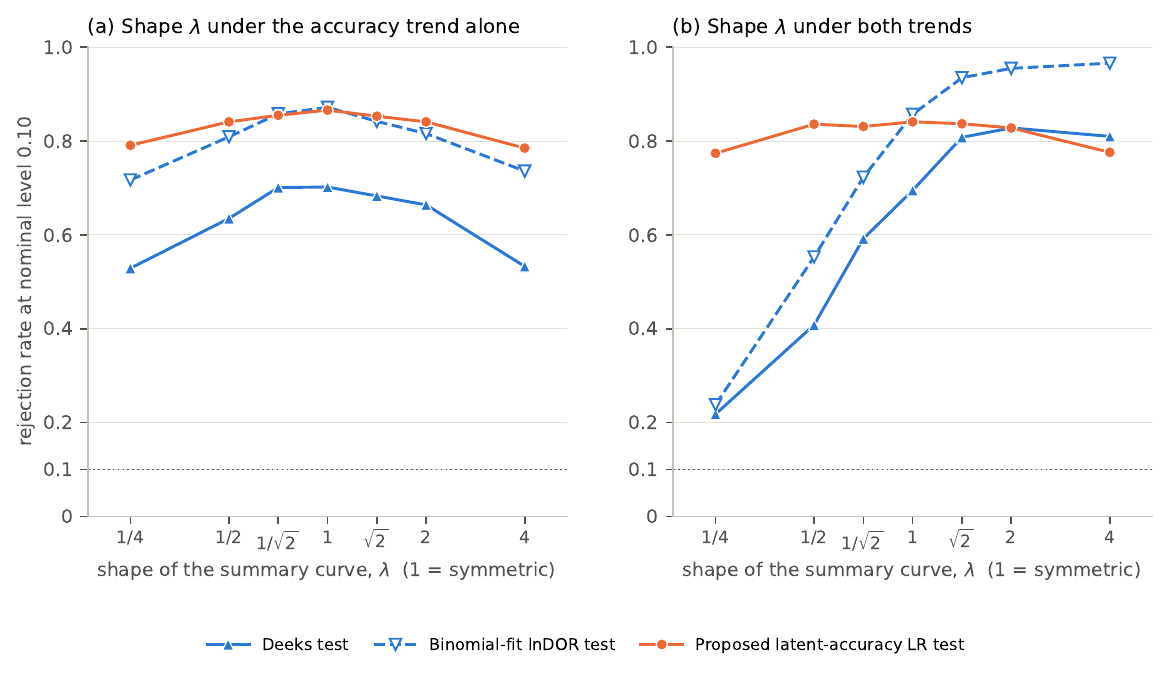}
\caption{Rejection under a fixed latent accuracy trend of $0.5$ per standard deviation of $s_i$, with $k=30$ and $\rho=0.4$. (a)~Accuracy trend alone. (b)~The same accuracy trend with a threshold trend of $0.4$ residual standard deviations per standard deviation of $s_i$. Shape $\lambda$ varies on a logarithmic axis. Each point uses 1000 replicates; the dotted line marks the nominal level $0.10$. Colours and markers are as in Figure~\ref{fig:sim}. Rates are in Table~S1 of the supplementary material.}
\label{fig:power}
\par\vspace{0.3em}\begin{flushleft}\footnotesize Alt text: Two panels compare the same accuracy trend without and with a threshold trend. In the left panel, all three curves peak near shape one. In the right panel, the two blue ln DOR curves fall below their left-panel values at shapes below one and rise above them at shapes above one. The orange accuracy-test curve changes little between panels.\end{flushleft}
\end{figure}

With an accuracy trend alone, the two binomial procedures had similar power at the symmetric shape, whereas the proposed test had higher power at the most asymmetric shapes (Figure~\ref{fig:power}(a)). Across the shapes examined, power ranged from $0.785$ to $0.866$ for the proposed test, from $0.529$ to $0.702$ for Deeks and from $0.717$ to $0.872$ for the binomial-fit $\lnDOR$ test. At $\lambda=1$, the two binomial procedures had power of $0.866$ for the proposed test and $0.872$ for the $\lnDOR$ test. The largest differences between them occurred at $\lambda=1/4$ and $4$.

Adding the fixed threshold trend changed the $\lnDOR$ tests in opposite directions on either side of $\lambda=1$ (Figure~\ref{fig:power}(b)). Below one, threshold and accuracy contributions opposed each other; above one, they added. For the binomial-fit $\lnDOR$ test, rejection fell from $0.717$ to $0.238$ at $\lambda=1/4$ and rose from $0.736$ to $0.966$ at $\lambda=4$. The accuracy trend was unchanged: it was the combined $\lnDOR$ trend that became harder or easier to detect. The proposed test had power of $0.774$ to $0.841$, within $0.025$ of its value without the threshold trend at each shape.

The additional comparisons gave the same qualitative distinction between the two null hypotheses (Table~S3 of the supplementary material). Power of the proposed test increased with the number of studies. It also increased with $\rho$: in this design, larger $\rho$ means less residual variation in latent accuracy, making the fixed accuracy change easier to detect. Across all 51 settings satisfying the accuracy null, its rejection rate ranged from $0.085$ to $0.134$; the detailed comparisons are reported in Section~C.3 of the supplementary material.

\section{Application to two published reviews}\label{sec:apps}

\subsection{Data}\label{sec:appdata}

We re-analyse two published DTA meta-analyses to illustrate how fitted curve shape affects the relation between $\lnDOR$ and latent accuracy trends. The first is the Cochrane review of faecal immunochemical tests (FIT) for colorectal cancer in average-risk screening populations\cite{grobbee2022}. We used its ``reference standard: positive'' analysis. In this design, FIT-positive participants were referred for colonoscopy, and FIT-negative participants were followed for interval cancers through registries or clinical follow-up. We obtained the published $2\times2$ tables from the review's data file. An LLM assisted with transcribing these tables into the analysis data set (see Generative AI statement). Both authors independently checked every entry against the review's data file. The data comprise 23 entries from 22 studies. One study contributes two FIT brands evaluated in separate participant groups (Section~D.1 of the supplementary material). The entries contain 1,179 to 747,076 participants and 6 to 1,578 cancers.

The review was identified by screening Cochrane DTA reviews with published data files for analyses with at least 20 entries and an asymmetric fitted curve. Among the review's FIT analyses, this analysis was selected because it had the largest studies, the widest range of study sizes and the fewest empty cells.

In this FIT analysis, false negatives are interval cancers identified during follow-up, which ranged from one to more than three years across entries. The sensitivity denominator therefore depends on the follow-up of test-negative participants. Because cancer prevalence is low, effective sample size is close to four times the number of cancers. The size axis thus primarily reflects the number of cancers rather than the total number of participants. Further details are in Section~D.1 of the supplementary material.

The second example is the meta-analysis of impedance plethysmography (IPG) for symptomatic deep-vein thrombosis from a health-technology assessment, comprising 42 cohorts from 1972--2000\cite{goodacre2006}. The source reports sensitivities and specificities with exact confidence intervals but not the group sizes. An LLM assisted with extracting the reported sensitivities, specificities and confidence intervals (see Generative AI statement). Both authors independently checked every extracted value against the source report. We then used our scripts to reconstruct integer $2\times2$ tables compatible with the printed values. The primary analysis used the minimum-count reconstruction within a restricted candidate set; we repeated the analysis with the maximum-count reconstruction from that set (Section~D.2 of the supplementary material). Across the two primary analyses, 6 of the 65 entries had an empty cell.

\subsection{Results}\label{sec:appresults}

\begin{table}[t]
\centering
\caption{Size trends in the two published reviews; IPG uses the minimum-count reconstruction. Trend estimates are per unit of $s_i=1/\sqrt{\ESS_i}$. The $\lnDOR$ trend is decomposed into threshold and accuracy contributions by Proposition~\ref{thm:decomp}(iii). Parentheses give standard errors. Deeks and the binomial-fit $\lnDOR$ test use a $t_{k-2}$ reference; latent-trend tests use the signed likelihood-ratio root with the same reference. Full fitted parameters and both IPG reconstruction endpoints are in Table~S9 of the supplementary material.}
\label{tab:apps}
\footnotesize
\setlength{\tabcolsep}{6pt}
\begin{tabular}{lcc}
\toprule
 & \makecell{FIT (colorectal cancer)\\$k=23$ entries} & \makecell{IPG (DVT)\\$k=42$ cohorts} \\
\midrule
\multicolumn{3}{l}{\emph{Size trends: estimate (SE); $p$}} \\
\addlinespace[2pt]
Deeks test: slope of $\lnDOR$ on $s_i$ & $-3.5$\ (4.1); 0.406 & $-10.7$\ (4.1); 0.012 \\
Fitted $\lnDOR$ trend, $\hat\beta_\eta-\hat\beta_\varphi$ & $-1.0$\ (2.7); 0.710 & $-9.7$\ (3.6); 0.010 \\
\quad of which threshold part; accuracy part & $+2.95$; $-3.94$ & $-0.02$; $-9.65$ \\
Proposed: latent accuracy trend $\hat\gamma_\alpha$ & $-3.7$\ (1.5); 0.029 & $-9.7$\ (3.9); 0.023 \\
\addlinespace[4pt]
\multicolumn{3}{l}{\emph{Threshold trend and curve shape}} \\
\addlinespace[2pt]
Latent threshold trend $\hat\gamma_\theta$ (SE); $p$ & $+3.9$\ (2.1); 0.081 & $+6.3$\ (2.3); 0.013 \\
Shape $\hat\lambda$ [95\% CI] & 2.11\ [1.73, 2.56] & 1.00\ [0.65, 1.53] \\
\bottomrule
\end{tabular}
\end{table}

\begin{figure}[p]
\centering
\includegraphics[width=0.98\textwidth]{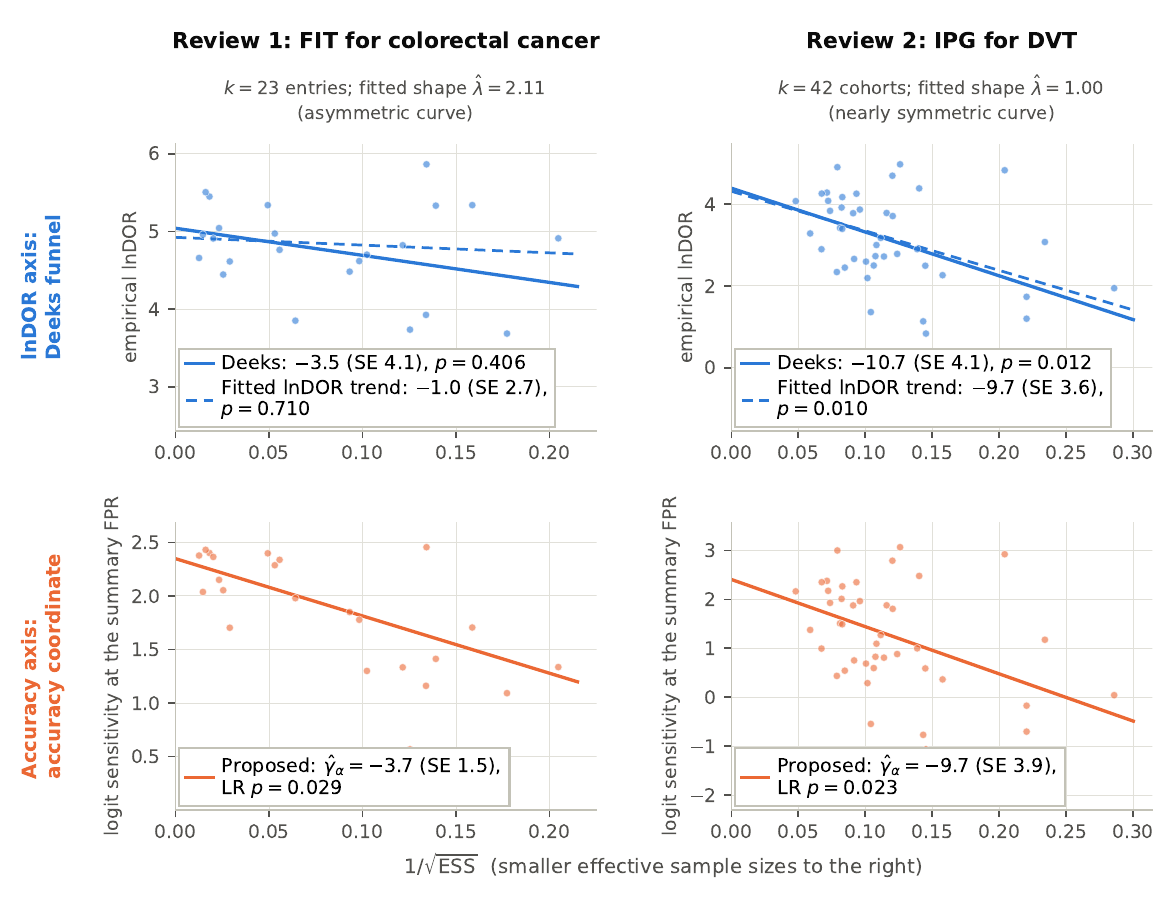}
\caption{Size trends in the FIT (left) and IPG (right) reviews. The upper panels show empirical $\lnDOR$ with the Deeks regression (solid blue) and binomial-fit $\lnDOR$ trend (dashed blue). The lower panels show the accuracy coordinate of Section~\ref{sec:inference} with its fitted trend. All panels use $s_i=1/\sqrt{\ESS_i}$, with smaller effective sample sizes to the right. Continuity correction is used for the displayed empirical logits and the Deeks analysis; the other lines and tests use the binomial likelihood. Numerical estimates are in Table~\ref{tab:apps}.}
\label{fig:apps}
\par\vspace{0.3em}\begin{flushleft}\footnotesize Alt text: Four funnel-type panels, one column per review and one row per axis. Left column, FIT: in the Deeks funnel a shallow solid line and a nearly flat dashed line; below it a clearly falling line in the accuracy coordinate. Right column, IPG: a falling solid line and a falling dashed line close together in the Deeks funnel above a falling line in the accuracy coordinate.\end{flushleft}
\end{figure}

The three tests gave different conclusions in FIT but agreed in IPG (Table~\ref{tab:apps} and Figure~\ref{fig:apps}). In both reviews, smaller effective sample sizes were associated with higher fitted latent thresholds and lower fitted latent accuracy. The fitted curve shape determined how these two associations combined on the $\lnDOR$ scale.

\emph{FIT for colorectal cancer.} The threshold and accuracy contributions largely cancelled on the $\lnDOR$ scale. The fitted shape was $\hat\lambda=2.11$ (95\% confidence interval 1.73 to 2.56). The threshold contribution to the $\lnDOR$ trend was $+2.95$, and the accuracy contribution was $-3.94$. Thus, the threshold contribution offset about three-quarters of the accuracy contribution. Neither $\lnDOR$ test detected a study-size trend (Deeks $p=0.406$; binomial-fit $\lnDOR$ $p=0.710$), whereas the proposed test detected a negative latent accuracy trend ($\hat\gamma_\alpha=-3.7$, $p=0.029$).

In the fitted model, smaller effective sample sizes were associated with more liberal operating positions and lower summary curves (Proposition~\ref{thm:curve}). In the observed data, the seven entries with the smallest effective sample sizes had an unweighted mean false-positive rate of 7.3\%, compared with 4.4\% in the seven largest. The difference between the proposed and Deeks tests persisted in all 23 leave-one-out fits: the proposed test rejected at the 10\% level in every fit and Deeks in none (Section~D.1 of the supplementary material).

\emph{IPG.} All three tests detected negative study-size trends. The fitted shape was close to one ($\hat\lambda=1.00$, 95\% confidence interval 0.65 to 1.53), and the threshold contribution to the $\lnDOR$ trend was negligible ($-0.02$). The $p$ values were 0.012 for Deeks, 0.010 for the binomial-fit $\lnDOR$ test and 0.023 for the proposed test. These conclusions were unchanged between the minimum-count and maximum-count reconstructions examined (Section~D.2 of the supplementary material).

\section{Discussion}\label{sec:disc}

We have developed a test for small-study effects in HSROC latent accuracy. Under the common-shape model, the accuracy null states that the same summary ROC curve applies across study sizes. The $\lnDOR$ null has this interpretation only at symmetry. With an asymmetric curve, a threshold trend can produce an $\lnDOR$ trend without a latent accuracy trend, or obscure a latent accuracy trend through opposing contributions.

In simulations, the proposed test remained close to its nominal level when latent accuracy had no size trend, including settings with a threshold trend. Its power changed little when a threshold trend was added to a fixed latent accuracy trend. The two reviews illustrated how curve shape affects the relation between the two assessments. In FIT, threshold and accuracy contributions largely cancelled on the $\lnDOR$ scale; in IPG, the nearly symmetric fitted curve made the threshold contribution negligible.

In FIT, the disagreement between the $\lnDOR$ and latent accuracy tests persisted when both used the same binomial fit. Replacing the empirical-logit regression with a binomial likelihood therefore did not resolve the difference between the two assessments. This illustrates why the null hypothesis matters alongside the estimation method. Related multivariate tests assess the joint null that both sensitivity and specificity funnels are symmetric\cite{hong2020,noma2020}. The proposed test instead concerns latent accuracy.

In both reviews, smaller effective sample sizes were associated with lower fitted latent accuracy. Such associations may reflect selective publication, differences in study quality or genuine clinical heterogeneity\cite{egger1997,sterne2011,deeks2005}. The proposed test does not distinguish these explanations. Selection models address the further question of how the summary curve changes under specified assumptions about selective publication\cite{zhou2023sroc,hu2026copas}.

The interpretation depends on the common-shape HSROC model. With one operating point per study, the fitted latent threshold direction cannot be assumed to represent the effect of changing a clinical cutoff within a study\cite{hamza2009}. Heterogeneity affecting the two logits differently can alter that direction, and a shape that varies with study size falls outside the curve-invariance result in Proposition~\ref{thm:curve}. Estimating shape can also be difficult with few studies or sparse tables\cite{takwoingi2017}. The finite-sample reference distribution used for the proposed test is an approximation.

The two applications also have data limitations. In FIT, observed false-positive rates ranged from $2.0\%$ to $11.5\%$, so interpretation beyond this range relies on extrapolation of the fitted curve. Variation in follow-up of test-negative participants can also affect observed sensitivity. In IPG, the conclusions were unchanged between the two reconstructions examined; this does not establish stability across all tables compatible with the published summaries.

The proposed procedure tests the null of no latent accuracy trend within the existing HSROC model\cite{rutter2001,harbord2007}. In practice, the latent accuracy and threshold trends should be reported together with the fitted shape and its uncertainty. Plotting the accuracy coordinate alongside the Deeks funnel helps interpret these trends. This assessment distinguishes systematic differences in the summary curve across study sizes from differences in operating position along it.

\section*{Software}

Analyses use Python (NumPy/SciPy). The public repository contains the fitting and testing code, simulation generator, study-level data, replicate-level results, and manifests of the 80 simulation settings and their seeds. It also includes the reconstruction scripts and the code used to produce the figures and tables. The repository is available at \url{https://github.com/Kazuki-Nishida/dta-small-study-effects}. The version referenced by this manuscript is package version \RepositoryVersion, commit \href{\RepositorySnapshot}{\texttt{\RepositoryCommit}} (2 October 2026), archived in Zenodo\cite{nonomiya2026software} at \url{https://doi.org/10.5281/zenodo.23333166}. The study-level data, stored analysis results and figures underlying this article are unchanged from version 1.0.1.

\section*{Supplementary material}

Supplementary material contains proofs, computational details, additional simulation results and diagnostics, and the data and sensitivity analyses for the two reviews.

\section*{Funding}

This work was supported by the Japan Society for the Promotion of Science (JSPS) KAKENHI [grant number JP20K23222 to K.N.] and the SECOM Science and Technology Foundation (SECOM Challenge 2025, to K.N.). The Foundation's role was limited to financial support.

\section*{Acknowledgements}

The authors gratefully acknowledge research funding from the SECOM Science and Technology Foundation.

\section*{Author contributions}

\noindent\textbf{Yuta Nonomiya:} Conceptualization; Data curation; Validation; Writing -- review \& editing.

\medskip\noindent\textbf{Kazuki Nishida:} Conceptualization; Data curation; Formal analysis; Funding acquisition; Writing -- original draft.

Both authors contributed to the extraction and curation of the data used in the clinical examples. All authors read and approved the final manuscript.

\section*{Competing interests}

Y.N. declares no competing interests. K.N. received research funding from the SECOM Science and Technology Foundation, as disclosed in the Funding statement. K.N. declares no other competing interests.

\section*{Data availability}

The study-level tables of the FIT review are the published counts of the Cochrane review's analysis\cite{grobbee2022}, and those of the IPG review were reconstructed from the published tables of its source\cite{goodacre2006} as described in Section~D of the supplementary material. Both datasets, the reconstruction scripts, analysis code and stored results are publicly archived in Zenodo at \url{https://doi.org/10.5281/zenodo.23333166} (package version \RepositoryVersion; \href{\RepositorySnapshot}{commit \texttt{\RepositoryCommit}}) and are also available at \url{https://github.com/Kazuki-Nishida/dta-small-study-effects}. Data provenance and reuse conditions are documented in the repository and in Section~D of the supplementary material.

\section*{Ethics and consent}

This study uses simulated data and study-level aggregate data from published reviews. No individual participant data were collected or accessed for this study. Ethics approval, consent to participate and consent for publication are not applicable.

\section*{Generative AI statement}

We made limited use of large language models (OpenAI ChatGPT, GPT-6 Astra Pro; Anthropic Claude, Fable 5.1), confining them to supporting tasks: English-language editing, code development and documentation, literature and source-data searches, checking our derivations and references for errors, and assistance in extracting and reconstructing study-level tables from published reports. We reviewed every LLM output, and both of us independently verified each extracted or reconstructed table against its published source. The research question, estimands, statistical methodology, proofs, simulation design and interpretation of the results are entirely our own, and we take full responsibility for the content of the manuscript.

\clearpage
\singlespacing
\bibliographystyle{rsm-ama}
\bibliography{references}

\clearpage
\onehalfspacing
\setcounter{section}{0}
\setcounter{table}{0}
\setcounter{figure}{0}
\setcounter{equation}{0}
\renewcommand{\thesection}{\Alph{section}}
\renewcommand{\thetable}{S\arabic{table}}
\renewcommand{\thefigure}{S\arabic{figure}}
\renewcommand{\theHsection}{supp.\Alph{section}}
\renewcommand{\theHsubsection}{supp.\Alph{section}.\arabic{subsection}}
\renewcommand{\theHtable}{supp.\arabic{table}}
\renewcommand{\theHfigure}{supp.\arabic{figure}}
\renewcommand{\theHequation}{supp.\Alph{section}.\arabic{equation}}
\section*{Supplementary Materials}

\noindent This supplement contains the proofs of the results stated in the main text (Section~A); the details of the likelihood fit, the constrained fits and the standard errors (Section~B); the simulation results and computational diagnostics (Section~C); and the data and sensitivity analyses for the two reviews (Section~D). Equation, proposition, table and figure numbers of the form (2.3), Proposition~2.1, Table~2 and Figure~1 refer to the main text; numbers of the form (A.1) and Table~S1 refer to this supplement.

\section{Proofs}\label{sapp:proofs}

\noindent Notation follows the main text: (2.2) is the HSROC parametrization with latent threshold $\theta_i$, latent accuracy $\alpha_i$, size trends $\gamma_\theta$, $\gamma_\alpha$ and shape $\bH$; the bivariate parametrization of Section~2.1 has trend vector $(\beta_\eta,\beta_\varphi)$ and residual standard deviations $\sigma_\eta,\sigma_\varphi$ with correlation $\rho$; $\lambda=e^{-\bH}=\sigma_\eta/\sigma_\varphi$; and $e_\pm=e^{\pm\bH/2}$, so that $e_-=\lambda^{1/2}$ and $e_+=\lambda^{-1/2}$. All moments are conditional on $s_i$.

\subsection{Proposition 2.1}\label{sapp:p21}

\begin{thmdecomp}
See the main text.
\end{thmdecomp}

\begin{proof}
(i) From (2.2), $\partial(\varphi_i,\eta_i)/\partial\theta_i=(e_+,e_-)\propto(1,e_-/e_+)=(1,\lambda)$ and $\partial(\varphi_i,\eta_i)/\partial\alpha_i=(-e_+/2,e_-/2)\propto(-1,\lambda)$. For the elimination, $\eta_i-\lambda\varphi_i=e_-(\theta_i+\alpha_i/2)-\lambda e_+(\theta_i-\alpha_i/2)$, and $\lambda e_+=e_-$, so the expression equals $e_-\alpha_i=\lambda^{1/2}\alpha_i$; likewise $\eta_i+\lambda\varphi_i=2e_-\theta_i=2\lambda^{1/2}\theta_i$:
\begin{equation}\label{eq:elimref}
\eta_i-\lambda\varphi_i=\lambda^{1/2}\alpha_i,\qquad \eta_i+\lambda\varphi_i=2\lambda^{1/2}\theta_i .
\end{equation}
(ii) Taking conditional expectations given $x_i$ in (2.2), $\beta_\eta=e_-(\gamma_\theta+\gamma_\alpha/2)$ and $\beta_\varphi=e_+(\gamma_\theta-\gamma_\alpha/2)$, which are the stated forms; the inverse follows by taking conditional expectations of \eqref{eq:elimref}, or directly by solving the two linear equations. The accuracy null $\gamma_\alpha=0$ is $\lambda^{-1/2}(\beta_\eta-\lambda\beta_\varphi)=0$.
(iii) Substituting (ii), $\beta_\eta-c\beta_\varphi=\lambda^{1/2}(\gamma_\theta+\gamma_\alpha/2)-c\lambda^{-1/2}(\gamma_\theta-\gamma_\alpha/2)=\lambda^{-1/2}[(\lambda-c)\gamma_\theta+\tfrac12(\lambda+c)\gamma_\alpha]$. The coefficient of $\gamma_\theta$ vanishes identically iff $c=\lambda$, where $\beta(\lambda)=\lambda^{-1/2}\lambda\gamma_\alpha=\lambda^{1/2}\gamma_\alpha$.
(iv) At $c=1$ the coefficient of $\gamma_\theta$ is $\lambda^{-1/2}(\lambda-1)$, zero iff $\lambda=1$ iff $\bH=0$. At $\bH=0$, $e_\pm=1$ and (2.2) reads $\eta_i=\theta_i+\alpha_i/2$, $\varphi_i=\theta_i-\alpha_i/2$, whence $\alpha_i=\eta_i-\varphi_i$ and $\theta_i=(\eta_i+\varphi_i)/2$.
\end{proof}

\subsection{Corollary 2.2}\label{sapp:c22}

\begin{corsize}
Under a pure threshold trend ($\gamma_\alpha=0$, $\gamma_\theta\neq0$) with $\bH\neq0$, the size trend of the $\lnDOR$ of the true logits is $\beta(1)=\lambda^{-1/2}(\lambda-1)\gamma_\theta\neq0$, and a two-sided $t$-type test based on an estimator of $\beta(1)$ that is consistent with estimated standard error $O_p(k^{-1/2})$ rejects the hypothesis of no $\lnDOR$ trend with probability tending to one as $k$ grows.
\end{corsize}

\begin{proof}
Under the bivariate parametrization, $E\{a_i(c)\mid s_i\}=(\mu_\eta-c\mu_\varphi)+(\beta_\eta-c\beta_\varphi)x_i$ is exactly linear in $s_i$. When the conditional mean is exactly linear, the population slope functional of a weighted least-squares regression of $a_i(c)$ on $s_i$ equals its slope for \emph{any} fixed weighting scheme, weights that depend on $s_i$ included, because the weighted normal equations are solved by the true conditional-mean coefficients. By Proposition~2.1(iii) with $c=1$ and $\gamma_\alpha=0$, that slope is $\lambda^{-1/2}(\lambda-1)\gamma_\theta\neq0$. An estimator consistent for it with estimated standard error $O_p(k^{-1/2})$ therefore has a $t$ statistic whose absolute value diverges, which gives the claim.
\end{proof}

\subsection{Proposition 2.3}\label{sapp:p23}

\begin{thmcurve}
Under (2.2), with $\bH$ not depending on $x$, $\partial S_x(u)/\partial x=\lambda^{1/2}\gamma_\alpha S_x(u)\{1-S_x(u)\}$ for every $u\in(0,1)$. Consequently the accuracy null holds iff the summary curve is the same for studies of all sizes, and under the alternative the sign of $\gamma_\alpha$ is the sign of the size trend of the sensitivity at every false-positive rate and of the area under the summary curve, $\SAUC(x)=\int_0^1S_x(u)\,du$.
\end{thmcurve}

\begin{proof}
By \eqref{eq:elimref}, a study with latent accuracy $\alpha$ lies on $\eta=\lambda\varphi+\lambda^{1/2}\alpha$, so the curve at mean accuracy $\alpha(x)=\Lambda+\gamma_\alpha x$ is $\logit S_x(u)=\lambda\logit u+\lambda^{1/2}\alpha(x)$. Differentiating in $x$ with $\lambda$ fixed, $\partial\logit S_x(u)/\partial x=\lambda^{1/2}\gamma_\alpha$, and since $dS/d\logit S=S(1-S)$, (2.8) follows. As $S_x(u)\{1-S_x(u)\}>0$ on $(0,1)$, the sign of $\partial S_x(u)/\partial x$ is the sign of $\gamma_\alpha$ for every $u$, and integrating a function of one sign over $(0,1)$ gives the statement for $\SAUC(x)$. If $\gamma_\alpha=0$ the curve does not depend on $x$; if $\gamma_\alpha\neq0$ the curves at two sizes differ at every $u$.
\end{proof}

\section{Estimation and inference}\label{sapp:comp}

\subsection{The size-adjusted bivariate binomial meta-regression}\label{sapp:glmm}

The fit is the generalized linear mixed model of Chu and Cole\cite{chu2006} with a size covariate: $\mathrm{TP}_i\sim\mathrm{Bin}\{n_{1i},\expit(\eta_i)\}$, $\mathrm{FP}_i\sim\mathrm{Bin}\{n_{0i},\expit(\varphi_i)\}$, and the bivariate normal random effects of Section~2.1 of the main text, with parameter vector $\vartheta=(\mu_\eta,\beta_\eta,\mu_\varphi,\beta_\varphi,\log\sigma_\eta,\log\sigma_\varphi,\operatorname{artanh}\rho)$. Each study's likelihood contribution is the two-dimensional integral of the product of its two binomial probabilities against the bivariate normal density, evaluated by adaptive Gauss--Hermite quadrature: for every study and every candidate $\vartheta$ the mode of the log integrand is found by damped Newton iteration (the integrand is log-concave, so the iteration converges from the prior mean), the integrand is re-centred at the mode and rescaled by the Cholesky factor of the inverse negative Hessian there, and a $9\times9$ product Gauss--Hermite rule is applied; a $15\times15$ rule changes no reported quantity in the third decimal. The likelihood is maximized by L-BFGS-B (bounds $\log\sigma\in[-4,3]$, $\operatorname{artanh}\rho\in[-4,4]$) from a normal-approximation fit of the empirical logits with size-only (smoothed) within-study variances as starting value, followed by a Nelder--Mead polish in the applications. The observed information is the numerical Hessian of the negative log-likelihood at the optimum, and $\widehat V$ its inverse.

\subsection{Constrained maximization for the likelihood-ratio tests}\label{sapp:lrt}

The accuracy null $\gamma_\alpha=0$ is imposed by reparametrizing the two trend coefficients through Proposition~2.1(ii) with $\gamma_\alpha=0$,
\[
\beta_\eta=\lambda^{1/2}\gamma_\theta,\qquad \beta_\varphi=\lambda^{-1/2}\gamma_\theta,\qquad \lambda=\exp(\log\sigma_\eta-\log\sigma_\varphi),
\]
so that the constrained model has the six parameters $(\mu_\eta,\mu_\varphi,\gamma_\theta,\log\sigma_\eta,\log\sigma_\varphi,\operatorname{artanh}\rho)$ and the same likelihood function; it is maximized by L-BFGS-B from two starts, the full-model estimate of $\gamma_\theta$ and zero, and the better optimum is kept. The threshold null $\gamma_\theta=0$ uses $\beta_\eta=\lambda^{1/2}\gamma_\alpha/2$ and $\beta_\varphi=-\lambda^{-1/2}\gamma_\alpha/2$. Each likelihood-ratio statistic is twice the difference of the maximized log-likelihoods; its asymptotic reference is $\chi^2_1$, and the main text refers its signed root to $t_{k-2}$, the small-sample convention of the Deeks test, a choice examined in Section~\ref{sapp:sim}. Because the constraint is written in the HSROC parametrization and the likelihood is that of the bivariate parametrization, the statistic is the same whichever parametrization is used to fit the full model\cite{harbord2007}, and it does not depend on whether the null is expressed for $\gamma_\alpha$ or for $\lambda^{1/2}\gamma_\alpha=\beta_\eta-\lambda\beta_\varphi$.

\subsection{Standard errors}\label{sapp:delta}

With $\hat\lambda=\exp(\log\hat\sigma_\eta-\log\hat\sigma_\varphi)$, the gradient of $\hat\lambda$ with respect to $(\log\sigma_\eta,\log\sigma_\varphi,\operatorname{artanh}\rho)$ is $\hat\lambda(1,-1,0)$. The latent trends $\hat\gamma_\alpha=\hat\lambda^{-1/2}\hat\beta_\eta-\hat\lambda^{1/2}\hat\beta_\varphi$ and $\hat\gamma_\theta=\tfrac12(\hat\lambda^{-1/2}\hat\beta_\eta+\hat\lambda^{1/2}\hat\beta_\varphi)$ are functions of the fitted trend vector and of $\hat\lambda$; their standard errors are $(g^\top\widehat Vg)^{1/2}$ with $g$ the gradient of the function with respect to $\vartheta$, in the order of $\vartheta$,
\[
g_\alpha=\bigl(0,\;\hat\lambda^{-1/2},\;0,\;-\hat\lambda^{1/2},\;-\hat\gamma_\theta,\;\hat\gamma_\theta,\;0\bigr),\qquad
g_\theta=\bigl(0,\;\tfrac12\hat\lambda^{-1/2},\;0,\;\tfrac12\hat\lambda^{1/2},\;-\tfrac14\hat\gamma_\alpha,\;\tfrac14\hat\gamma_\alpha,\;0\bigr),
\]
where the entries on $\log\sigma_\eta$ and $\log\sigma_\varphi$ come from the dependence of the functions on $\hat\lambda$: since $\partial\hat\lambda^{\pm1/2}/\partial\log\sigma_\eta=\pm\tfrac12\hat\lambda^{\pm1/2}$, $\partial\hat\gamma_\alpha/\partial\log\sigma_\eta=-\tfrac12\hat\lambda^{-1/2}\hat\beta_\eta-\tfrac12\hat\lambda^{1/2}\hat\beta_\varphi=-\hat\gamma_\theta$ and $\partial\hat\gamma_\theta/\partial\log\sigma_\eta=-\tfrac14\hat\gamma_\alpha$, with the opposite signs for $\log\sigma_\varphi$; these entries carry the estimation variance of the shape and its covariance with the two slope estimates into the standard error. The shape $\hat\bH=\log\hat\sigma_\varphi-\log\hat\sigma_\eta$ has standard error $(\widehat V_{55}+\widehat V_{66}-2\widehat V_{56})^{1/2}$ (indices in the order of $\vartheta$), and the interval for $\lambda$ is the exponential of minus the $t_{k-2}$ interval for $\bH$. The fitted trend of the accuracy coordinate displayed in Figure~3 of the main text is $\hat\beta_\eta-\hat\lambda\hat\beta_\varphi=\hat\lambda^{1/2}\hat\gamma_\alpha$ per unit of $s_i$, drawn through $\hat\mu_\eta$ at $x=0$.

\section{Simulation design and performance}\label{sapp:sim}

\subsection{Generator and settings}\label{sapp:gen}

Study sizes are drawn once per replicate: total size $N_i$ log-normal with median 300 and log-scale standard deviation $0.8$, rounded and clipped to $[40,4000]$, with $n_{1i}=\max(\lfloor0.35N_i\rfloor,10)$ diseased and $n_{0i}=\max(N_i-n_{1i},10)$ non-diseased; $\ESS_i=4n_{1i}n_{0i}/(n_{1i}+n_{0i})$, $s_i=1/\sqrt{\ESS_i}$, and $z_i=(s_i-\bar s)/\mathrm{SD}(s)$ is $s_i$ standardized within the replicate (population standard deviation). With $\lambda^{1/2}=e^{-\bH/2}$, the generator sets $\Theta=(\mu_\eta\lambda^{-1/2}+\mu_\varphi\lambda^{1/2})/2$ and $\Lambda=\mu_\eta\lambda^{-1/2}-\mu_\varphi\lambda^{1/2}$ so that the mean operating point is $(\mu_\eta,\mu_\varphi)=(1,-2)$ at every shape, $\sigma_\theta^2=\sigma_\eta\sigma_\varphi(1+\rho)/2$ and $\sigma_\alpha^2=2\sigma_\eta\sigma_\varphi(1-\rho)$ with $\sigma_\eta\sigma_\varphi=0.75^2$ (the correspondence of Harbord et al.\cite{harbord2007} between the two parametrizations), and generates
\begin{gather*}
\theta_i=\Theta+\rho_s\,\sigma_\theta z_i+\sigma_\theta\varepsilon_i,\qquad
\alpha_i=\Lambda+\delta\,z_i+\sigma_\alpha\varepsilon_i',\\
\eta_i=\lambda^{1/2}(\theta_i+\alpha_i/2),\qquad \varphi_i=\lambda^{-1/2}(\theta_i-\alpha_i/2),
\end{gather*}
with independent standard normal $\varepsilon_i,\varepsilon_i'$; the counts are binomial given $(\eta_i,\varphi_i)$, with $\expit(\varphi_i)$ bounded below by $10^{-4}$. The threshold trend is $\rho_s$ residual standard deviations of the latent threshold per standard deviation of $s_i$ and the accuracy trend $\delta$ per standard deviation of $s_i$, both on the generating scale $z_i$; the fits use $x_i=s_i-\bar s$, so the fitted trends relate to the generating ones by $\gamma_\theta\,\mathrm{SD}(s)=\rho_s\sigma_\theta$ and $\gamma_\alpha\,\mathrm{SD}(s)=\delta$, with $\mathrm{SD}(s)$ varying between replicates. Because the trends are added on top of the residual variation, the residual $\lambda$ and $\rho$ equal the design values in every scenario; $\rho$ is the residual correlation of the two logits induced by the HSROC parametrization, and the latent threshold and latent accuracy residuals are independent at every $\rho$, so varying $\rho$ does not relax that independence.

\emph{Baseline settings} (Figures~1 and 2 of the main text): $k=30$ and $\rho=0.4$; the four scenarios $(\rho_s,\delta)=(0,0)$, $(0.4,0)$, $(0,0.5)$ and $(0.4,0.5)$ at the seven shapes $\lambda\in\{1/4,1/2,1/\sqrt2,1,\sqrt2,2,4\}$ (Table~\ref{tab:sim_baseline}; Figures~1(a), 2(a) and 2(b)); the threshold-trend strengths $\rho_s\in\{0,0.2,\ldots,1.6\}$ at $\lambda=1/2$ with $\delta=0$ (Figure~1(b); Table~\ref{tab:sim_threshold}). For the shape comparison with both trends present, the true $\lnDOR$ trend per standard deviation of study size is given by
\begin{equation}\label{eq:bdor}
b_{\mathrm{DOR}}=\lambda^{-1/2}\Bigl[(\lambda-1)\rho_s\sigma_\theta+\frac{\lambda+1}{2}\,\delta\Bigr],
\end{equation}
the generating counterpart of Proposition~2.1(iii). \emph{Additional settings} (Table~\ref{tab:sim_krholambda}): the number of studies was varied at $\rho=0.4$ ($k\in\{10,20,30,50\}$), and the residual correlation was varied at $k=30$ ($\rho\in\{0,0.4,0.8\}$). Each comparison used three shapes, $\lambda\in\{1/2,1,2\}$, and the three scenarios specified in Table~\ref{tab:sim_krholambda} (no trend; a threshold trend alone, $\rho_s=0.4$; an accuracy trend alone, $\delta=0.5$). The resulting 54 settings include nine baseline settings, leaving 45 additional settings; the 35 baseline settings are the 28 of the shape comparisons and the seven further strengths of Table~\ref{tab:sim_threshold}. With $\rho$ the residual standard deviations of the latent variables change: $\sigma_\alpha=1.061$, $0.822$ and $0.474$ and $\sigma_\theta=0.530$, $0.627$ and $0.712$ at $\rho=0$, $0.4$ and $0.8$, so $\delta=0.5$ is $0.47$, $0.61$ and $1.05$ residual standard deviations of the latent accuracy: the same change of the latent accuracy at every $\rho$, not the same standardized effect (the tables give $\delta/\sigma_\alpha$). The reported design therefore comprises 80 distinct settings, with 1000 replicates per setting. The same generator and the same seeds are used throughout: every setting has its own seed, so the replicates of different settings are independent (no common random numbers), and the replicate data sets and the seeds are stored with the code, together with a manifest listing the 80 simulation settings and their seeds.

Every replicate is fitted as in Section~\ref{sapp:glmm} (L-BFGS-B from the normal-approximation start, numerical Hessian, no polish) with the constrained fit of Section~\ref{sapp:lrt}; the Deeks test uses the continuity-corrected empirical logits ($+0.5$ in every cell) with $\ESS$ weights. Replicates in which the optimizer reported non-convergence are retained; a $p$ value counts as valid when it is finite and, for the Wald contrasts, when the variance used is finite and positive, and rejection rates are computed over the valid replicates (their number is 1000 unless Tables~\ref{tab:sim_checks} to \ref{tab:sim_checks_d} say otherwise). With 1000 replicates, the Monte Carlo standard error of a rejection rate near $0.10$ is approximately $0.0095$ and is at most $0.016$ for any rate.

\subsection{Performance at the baseline settings}\label{sapp:grid}

Table~\ref{tab:sim_baseline} reports the baseline rejection rates for the Deeks test, the binomial-fit $\lnDOR$ test and the proposed accuracy test. Table~\ref{tab:sim_threshold} reports the threshold-strength sweep. Table~\ref{tab:sim_checks} gives auxiliary statistics and computational diagnostics from the same fits.

Two auxiliary comparisons examine the inference used in the main text. In the no-trend baseline settings, referring the signed likelihood-ratio root to $t_{k-2}$ brought rejection closer to $0.10$ than the $\chi^2_1$ reference at every shape. The accuracy-coordinate Wald test and the proposed likelihood-ratio test differed by at most $0.013$ in the no-trend, threshold-only and accuracy-only scenarios on the baseline shape grid. Their maximum difference was $0.008$ on the threshold-only strength grid and $0.017$ in the both-trend shape comparison. These are comparisons of test forms on the same fits; their reference distributions and diagnostics are documented in Section~\ref{sapp:checks}.

In the both-trend comparison, equation~\eqref{eq:bdor} gives a true $\lnDOR$ trend of $+0.249$ at $\lambda=1/4$, $+0.500$ at $\lambda=1$ and $+1.001$ at $\lambda=4$. The accuracy contribution alone is $+0.625$, $+0.500$ and $+0.625$, respectively. The threshold contribution changes sign at one, explaining the opposite responses on the two sides of Figure~2(b).

\begin{table}[p]
\centering
\caption{Rejection rates at the nominal level $0.10$ in the baseline settings, $k=30$ and $\rho=0.4$, with 1000 valid replicates per cell. These include the data behind Figures~1(a), 2(a) and 2(b). Deeks, the binomial-fit $\lnDOR$ test and the proposed accuracy test are defined in Section~2.4 of the main text. $\delta$ is the latent accuracy change per standard deviation of $s_i$; $\rho_s$ is the threshold change in residual standard deviations per standard deviation of $s_i$. The final block also gives the true $\lnDOR$ trend $b_{\mathrm{DOR}}$ from equation~\eqref{eq:bdor}.}
\label{tab:sim_baseline}
\scriptsize
\setlength{\tabcolsep}{5pt}
\begin{tabular}{@{}ll ccc@{}}
\toprule
Scenario & $\lambda$ & Deeks & $\lnDOR$ fit & Proposed \\
\midrule
 No trend & $1/4$ & 0.153 & 0.124 & 0.117 \\
  & $1/2$ & 0.156 & 0.102 & 0.101 \\
  & $1/\sqrt2$ & 0.137 & 0.107 & 0.109 \\
  & $1$ & 0.157 & 0.120 & 0.117 \\
  & $\sqrt2$ & 0.154 & 0.098 & 0.105 \\
  & $2$ & 0.137 & 0.111 & 0.121 \\
  & $4$ & 0.138 & 0.121 & 0.120 \\
\addlinespace
 Threshold trend ($\rho_s=0.4$) & $1/4$ & 0.357 & 0.420 & 0.099 \\
  & $1/2$ & 0.240 & 0.263 & 0.115 \\
  & $1/\sqrt2$ & 0.188 & 0.146 & 0.115 \\
  & $1$ & 0.136 & 0.099 & 0.098 \\
  & $\sqrt2$ & 0.170 & 0.151 & 0.100 \\
  & $2$ & 0.238 & 0.229 & 0.115 \\
  & $4$ & 0.311 & 0.399 & 0.133 \\
\addlinespace
 Accuracy trend ($\delta=0.5$) & $1/4$ & 0.529 & 0.717 & 0.791 \\
  & $1/2$ & 0.635 & 0.809 & 0.841 \\
  & $1/\sqrt2$ & 0.701 & 0.858 & 0.855 \\
  & $1$ & 0.702 & 0.872 & 0.866 \\
  & $\sqrt2$ & 0.683 & 0.842 & 0.853 \\
  & $2$ & 0.664 & 0.816 & 0.841 \\
  & $4$ & 0.533 & 0.736 & 0.785 \\
\addlinespace
 Both trends ($\rho_s=0.4$, $\delta=0.5$) & $1/4$ ($b_{\mathrm{DOR}}=+0.249$) & 0.218 & 0.238 & 0.774 \\
  & $1/2$ ($+0.353$) & 0.408 & 0.553 & 0.836 \\
  & $1/\sqrt2$ ($+0.420$) & 0.592 & 0.723 & 0.831 \\
  & $1$ ($+0.500$) & 0.695 & 0.857 & 0.841 \\
  & $\sqrt2$ ($+0.595$) & 0.808 & 0.935 & 0.837 \\
  & $2$ ($+0.708$) & 0.828 & 0.955 & 0.828 \\
  & $4$ ($+1.001$) & 0.810 & 0.966 & 0.776 \\
\bottomrule
\end{tabular}
\end{table}

\begin{table}[p]
\centering
\caption{Rejection rates at the nominal level $0.10$ as threshold-trend strength $\rho_s$ varies at $k=30$, $\rho=0.4$ and $\lambda=1/2$, with no accuracy trend ($\delta=0$; 1000 replicates per setting). These are the data behind Figure~1(b). Columns are as in Table~\ref{tab:sim_baseline}.}
\label{tab:sim_threshold}
\scriptsize
\setlength{\tabcolsep}{5pt}
\begin{tabular}{@{}l ccc@{}}
\toprule
$\rho_s$ & Deeks & $\lnDOR$ fit & Proposed \\
\midrule
 $0$ & 0.156 & 0.102 & 0.101 \\
 $0.2$ & 0.181 & 0.131 & 0.095 \\
 $0.4$ & 0.240 & 0.263 & 0.115 \\
 $0.6$ & 0.313 & 0.377 & 0.117 \\
 $0.8$ & 0.424 & 0.547 & 0.113 \\
 $1.0$ & 0.521 & 0.687 & 0.102 \\
 $1.2$ & 0.665 & 0.828 & 0.113 \\
 $1.4$ & 0.752 & 0.899 & 0.100 \\
 $1.6$ & 0.849 & 0.961 & 0.103 \\
\bottomrule
\end{tabular}
\end{table}

\subsection{Additional settings: number of studies, residual correlation and shape}\label{sapp:krl}

Table~\ref{tab:sim_krholambda} gives the three procedures in two one-factor comparisons around the baseline: the number of studies $k\in\{10,20,30,50\}$ at $\rho=0.4$, and the residual correlation $\rho\in\{0,0.4,0.8\}$ at $k=30$, each at $\lambda\in\{1/2,1,2\}$ under the three scenarios of the baseline (no trend; threshold trend $\rho_s=0.4$; accuracy trend $\delta=0.5$), with $\delta/\sigma_\alpha$ for each $\rho$. The nine settings at $k=30$, $\rho=0.4$ are baseline settings of Table~\ref{tab:sim_baseline} (identical seeds and rates) and are shown once; the other 45 settings are additional.

\emph{Number of studies varied} ($\rho=0.4$; 36 settings). With no trend the Deeks test rose above its level as $k$ grew, from $0.102$ to $0.113$ at $k=10$ to $0.170$ to $0.192$ at $k=50$, while the binomial-fit $\lnDOR$ test ($0.098$ to $0.123$) and the proposed test ($0.096$ to $0.134$, its largest value at $k=10$) stayed near their level at every $k$. Under the threshold trend the rejection rates of the two $\lnDOR$-axis procedures on the asymmetric curves grew with $k$, from $0.123$ to $0.157$ at $k=10$ to $0.257$ to $0.300$ at $k=50$, and stayed at their no-trend values at $\lambda=1$, while the proposed test stayed within $0.090$ to $0.131$. Under the accuracy trend the power of every procedure grew with $k$: the proposed test from $0.440$ to $0.488$ at $k=10$ to $0.963$ to $0.973$ at $k=50$, at least $0.077$ above the Deeks test and within $0.033$ of the $\lnDOR$ fit at every $k$ (within $0.01$ at $k=10$).

\emph{Residual correlation varied} ($k=30$; 27 settings). With no trend the rates changed little with $\rho$: the Deeks test $0.137$ to $0.157$, the $\lnDOR$ fit $0.090$ to $0.120$ and the proposed test $0.085$ to $0.121$ over the nine settings. Under the threshold trend the rejection rates of the two $\lnDOR$-axis procedures on the asymmetric curves grew with $\rho$, from $0.183$ to $0.212$ at $\rho=0$ to $0.309$ to $0.393$ at $\rho=0.8$, while the proposed test stayed within $0.096$ to $0.128$. Under the accuracy trend the power of every procedure grew with $\rho$, the same change of the latent accuracy being $0.47$, $0.61$ and $1.05$ of its residual standard deviation: the proposed test from $0.705$ to $0.726$ at $\rho=0$ to $0.964$ to $0.983$ at $\rho=0.8$, at least $0.035$ above the Deeks test and within $0.049$ of the $\lnDOR$ fit in every setting. Across the 51 reported settings satisfying the accuracy null, the proposed test rejected in $0.085$ to $0.134$ of the replicates (Tables~\ref{tab:sim_baseline} to \ref{tab:sim_krholambda}).

\begin{table}[p]
\centering
\caption{Rejection rates at the nominal level $0.10$ when $k$ varies at $\rho=0.4$ and when $\rho$ varies at $k=30$. Each comparison uses $\lambda=1/2$, $1$ and $2$ under neither trend, a threshold trend alone ($\rho_s=0.4$) and an accuracy trend alone ($\delta=0.5$), with 1000 replicates per setting. The shared reference rows are shown once. $\delta/\sigma_\alpha$ gives the accuracy change in residual standard deviations. D: Deeks; F: binomial-fit $\lnDOR$; P: proposed accuracy test. Auxiliary statistics and diagnostics are in Tables~\ref{tab:sim_checks} to \ref{tab:sim_checks_d}.}
\label{tab:sim_krholambda}
\scriptsize
\setlength{\tabcolsep}{5pt}
\begin{tabular}{@{}ccrc ccc ccc ccc@{}}
\toprule
 & & & & \multicolumn{3}{c}{No trend} & \multicolumn{3}{c}{Threshold trend ($\rho_s=0.4$)} & \multicolumn{3}{c}{Accuracy trend ($\delta=0.5$)} \\
\cmidrule(lr){5-7}\cmidrule(lr){8-10}\cmidrule(lr){11-13}
$k$ & $\rho$ & $\lambda$ & $\delta/\sigma_\alpha$ & D & F & P & D & F & P & D & F & P \\
\midrule
\multicolumn{13}{@{}l}{\emph{Number of studies varied at $\rho=0.4$}} \\
\addlinespace[2pt]
 10 & 0.4 & $1/2$ & 0.61 & 0.113 & 0.117 & 0.134 & 0.124 & 0.157 & 0.131 & 0.360 & 0.443 & 0.443 \\
  &  & $1$ &  & 0.111 & 0.112 & 0.113 & 0.122 & 0.123 & 0.117 & 0.411 & 0.496 & 0.488 \\
  &  & $2$ &  & 0.102 & 0.110 & 0.123 & 0.123 & 0.138 & 0.105 & 0.356 & 0.437 & 0.440 \\
\addlinespace[2pt]
 20 & 0.4 & $1/2$ & 0.61 & 0.143 & 0.110 & 0.121 & 0.209 & 0.192 & 0.112 & 0.515 & 0.677 & 0.710 \\
  &  & $1$ &  & 0.139 & 0.109 & 0.118 & 0.117 & 0.100 & 0.102 & 0.602 & 0.748 & 0.755 \\
  &  & $2$ &  & 0.144 & 0.123 & 0.113 & 0.190 & 0.194 & 0.119 & 0.498 & 0.663 & 0.690 \\
\addlinespace[2pt]
 30 & 0.4 & $1/2$ & 0.61 & 0.156 & 0.102 & 0.101 & 0.240 & 0.263 & 0.115 & 0.635 & 0.809 & 0.841 \\
  &  & $1$ &  & 0.157 & 0.120 & 0.117 & 0.136 & 0.099 & 0.098 & 0.702 & 0.872 & 0.866 \\
  &  & $2$ &  & 0.137 & 0.111 & 0.121 & 0.238 & 0.229 & 0.115 & 0.664 & 0.816 & 0.841 \\
\addlinespace[2pt]
 50 & 0.4 & $1/2$ & 0.61 & 0.192 & 0.104 & 0.105 & 0.257 & 0.291 & 0.105 & 0.809 & 0.955 & 0.965 \\
  &  & $1$ &  & 0.184 & 0.098 & 0.096 & 0.180 & 0.097 & 0.090 & 0.844 & 0.971 & 0.973 \\
  &  & $2$ &  & 0.170 & 0.107 & 0.125 & 0.284 & 0.300 & 0.093 & 0.798 & 0.944 & 0.963 \\
\addlinespace[5pt]
\multicolumn{13}{@{}l}{\emph{Residual correlation varied at $k=30$ (the reference rows $\rho=0.4$ are those of $k=30$ above)}} \\
\addlinespace[2pt]
 30 & 0 & $1/2$ & 0.47 & 0.145 & 0.117 & 0.107 & 0.183 & 0.184 & 0.102 & 0.532 & 0.713 & 0.715 \\
  &  & $1$ &  & 0.155 & 0.115 & 0.110 & 0.152 & 0.123 & 0.128 & 0.572 & 0.719 & 0.726 \\
  &  & $2$ &  & 0.149 & 0.090 & 0.085 & 0.212 & 0.190 & 0.105 & 0.522 & 0.698 & 0.705 \\
\addlinespace[2pt]
 30 & 0.8 & $1/2$ & 1.05 & 0.143 & 0.112 & 0.100 & 0.341 & 0.393 & 0.096 & 0.822 & 0.948 & 0.964 \\
  &  & $1$ &  & 0.152 & 0.108 & 0.108 & 0.142 & 0.116 & 0.112 & 0.948 & 0.985 & 0.983 \\
  &  & $2$ &  & 0.151 & 0.117 & 0.106 & 0.309 & 0.362 & 0.128 & 0.809 & 0.933 & 0.982 \\
\bottomrule
\end{tabular}
\end{table}

\subsection{Auxiliary statistics and computational diagnostics}\label{sapp:checks}

Tables~\ref{tab:sim_checks} to \ref{tab:sim_checks_d} collect, for every reported setting, the auxiliary statistics computed on the same fits and the computational diagnostics (the 35 baseline settings in Table~\ref{tab:sim_checks}; the 45 additional settings of Table~\ref{tab:sim_krholambda} in Tables~\ref{tab:sim_checks_b} to \ref{tab:sim_checks_d}). Wald$_t$ is the Wald contrast of the accuracy coordinate $\hat\beta_\eta-\hat\lambda\hat\beta_\varphi$ (Section~\ref{sapp:delta}) referred to $t_{k-2}$, Proposed$_{\chi^2}$ the proposed test with its likelihood-ratio statistic referred to $\chi^2_1$ instead of $t_{k-2}$. The diagnostics count, out of 1000 replicates, the fits in which the optimizer reported non-convergence (full fit; constrained fit), the full fits with a parameter at a bound of the optimizer ($\log\sigma$ at $-4$ or $3$, $\operatorname{artanh}\rho$ at $\pm4$), the replicates in which at least one $p$ value was invalid (a non-finite value, or a Wald variance that was not finite and positive), and the replicates in which the likelihood ratio was negative beyond rounding ($<-10^{-6}$). Over the 80 reported settings no replicate produced an invalid $p$ value for the three procedures or the auxiliary statistics, the numerical Hessian of the full fit was invertible in every replicate, the optimizer reported non-convergence in at most $2$ full fits and $4$ constrained fits per 1000, and no likelihood-ratio statistic below $-10^{-6}$ was recorded. Parameters at a bound of the optimizer are almost always $\operatorname{artanh}\hat\rho$ at $\pm4$, that is $|\hat\rho|\ge0.9993$ ($97\%$ of the bound hits in the reported settings): they occur in $3.7$ to $7.8\%$ of the full fits at $k=10$, in at most $0.2\%$ at $k=20$, $0.7\%$ at $k=30$ (the largest values at $\rho=0.8$) and $0.0\%$ at $k=50$; the rates of Tables~\ref{tab:sim_baseline} to \ref{tab:sim_krholambda} include these replicates. Referring the likelihood ratio to $\chi^2_1$ instead of $t_{k-2}$ raises the no-trend rejection rate by $0.037$ to $0.048$ at $k=10$, $0.011$ to $0.019$ at $k=20$, $0.006$ to $0.017$ at $k=30$ and $0.005$ to $0.011$ at $k=50$: the $t$ convention of the main text matters at small $k$ and little at $k=50$. The Wald contrast of the accuracy coordinate agrees with the likelihood-ratio form to within $0.022$ with no trend, $0.011$ under the threshold trends, $0.035$ under the accuracy trend alone (the largest gaps at $k=10$) and $0.017$ with both trends present at the baseline.

\begin{table}[p]
\centering
\caption{Auxiliary statistics and computational diagnostics for the 35 baseline settings ($k=30$, $\rho=0.4$; 1000 replicates per setting): the shape comparisons under no trend, a threshold trend and an accuracy trend, the further strengths of the threshold trend, and the shape comparison with both trends present. The columns report rejection rates at the nominal level $0.10$ of the accuracy-coordinate Wald test with the $t_{k-2}$ reference (Wald$_t$) and of the proposed likelihood-ratio statistic with the $\chi^2_1$ reference (Proposed$_{\chi^2}$), and the numerical diagnostics, counts out of 1000: NC, non-convergence reported by the optimizer (full fit / constrained fit); B, a parameter of the full fit at a bound; I, at least one invalid $p$ value; L, a negative likelihood ratio beyond rounding. Settings are identified by $\lambda$, $\rho_s$ and $\delta$.}
\label{tab:sim_checks}
\scriptsize
\setlength{\tabcolsep}{4pt}
\begin{tabular}{@{}ccc cc cccc@{}}
\toprule
$\lambda$ & $\rho_s$ & $\delta$ & Wald$_t$ & Proposed$_{\chi^2}$ & NC & B & I & L \\
\midrule
 $1/4$ & 0 & 0 & 0.120 & 0.130 & 0/2 & 0 & 0 & 0 \\
 $1/4$ & 0.4 & 0 & 0.100 & 0.114 & 0/0 & 0 & 0 & 0 \\
 $1/4$ & 0 & 0.5 & 0.778 & 0.805 & 0/1 & 1 & 0 & 0 \\
 $1/2$ & 0 & 0 & 0.109 & 0.115 & 0/0 & 0 & 0 & 0 \\
 $1/2$ & 0.4 & 0 & 0.114 & 0.131 & 0/2 & 0 & 0 & 0 \\
 $1/2$ & 0 & 0.5 & 0.836 & 0.851 & 0/1 & 0 & 0 & 0 \\
 $1/\sqrt2$ & 0 & 0 & 0.110 & 0.118 & 0/3 & 0 & 0 & 0 \\
 $1/\sqrt2$ & 0.4 & 0 & 0.116 & 0.121 & 0/2 & 0 & 0 & 0 \\
 $1/\sqrt2$ & 0 & 0.5 & 0.854 & 0.877 & 0/0 & 0 & 0 & 0 \\
 $1$ & 0 & 0 & 0.116 & 0.130 & 0/2 & 0 & 0 & 0 \\
 $1$ & 0.4 & 0 & 0.100 & 0.111 & 0/1 & 0 & 0 & 0 \\
 $1$ & 0 & 0.5 & 0.866 & 0.879 & 0/1 & 0 & 0 & 0 \\
 $\sqrt2$ & 0 & 0 & 0.105 & 0.116 & 0/0 & 0 & 0 & 0 \\
 $\sqrt2$ & 0.4 & 0 & 0.104 & 0.112 & 0/0 & 0 & 0 & 0 \\
 $\sqrt2$ & 0 & 0.5 & 0.847 & 0.862 & 0/1 & 0 & 0 & 0 \\
 $2$ & 0 & 0 & 0.120 & 0.137 & 0/0 & 0 & 0 & 0 \\
 $2$ & 0.4 & 0 & 0.113 & 0.128 & 0/3 & 0 & 0 & 0 \\
 $2$ & 0 & 0.5 & 0.841 & 0.854 & 0/1 & 0 & 0 & 0 \\
 $4$ & 0 & 0 & 0.114 & 0.135 & 0/4 & 1 & 0 & 0 \\
 $4$ & 0.4 & 0 & 0.124 & 0.143 & 0/0 & 3 & 0 & 0 \\
 $4$ & 0 & 0.5 & 0.775 & 0.796 & 0/1 & 0 & 0 & 0 \\
\addlinespace[2pt]
 $1/2$ & 0.2 & 0 & 0.098 & 0.110 & 0/0 & 0 & 0 & 0 \\
 $1/2$ & 0.6 & 0 & 0.113 & 0.131 & 0/2 & 0 & 0 & 0 \\
 $1/2$ & 0.8 & 0 & 0.111 & 0.127 & 0/0 & 1 & 0 & 0 \\
 $1/2$ & 1 & 0 & 0.106 & 0.114 & 0/1 & 0 & 0 & 0 \\
 $1/2$ & 1.2 & 0 & 0.115 & 0.125 & 1/0 & 0 & 0 & 0 \\
 $1/2$ & 1.4 & 0 & 0.104 & 0.116 & 0/0 & 0 & 0 & 0 \\
 $1/2$ & 1.6 & 0 & 0.106 & 0.113 & 1/1 & 0 & 0 & 0 \\
\addlinespace[2pt]
 $1/4$ & 0.4 & 0.5 & 0.785 & 0.792 & 0/0 & 1 & 0 & 0 \\
 $1/2$ & 0.4 & 0.5 & 0.840 & 0.844 & 0/1 & 0 & 0 & 0 \\
 $1/\sqrt2$ & 0.4 & 0.5 & 0.841 & 0.846 & 0/1 & 0 & 0 & 0 \\
 $1$ & 0.4 & 0.5 & 0.849 & 0.849 & 0/2 & 0 & 0 & 0 \\
 $\sqrt2$ & 0.4 & 0.5 & 0.848 & 0.849 & 2/1 & 0 & 0 & 0 \\
 $2$ & 0.4 & 0.5 & 0.834 & 0.834 & 0/3 & 0 & 0 & 0 \\
 $4$ & 0.4 & 0.5 & 0.793 & 0.795 & 0/0 & 1 & 0 & 0 \\
\bottomrule
\end{tabular}
\end{table}

\begin{table}[p]
\centering
\caption{Auxiliary statistics and computational diagnostics of the additional settings with $k=10$ and $\rho=0.4$ (1000 replicates per setting). Columns as in Table~\ref{tab:sim_checks}.}
\label{tab:sim_checks_b}
\scriptsize
\setlength{\tabcolsep}{4pt}
\begin{tabular}{@{}ccccc cc cccc@{}}
\toprule
$k$ & $\rho$ & $\lambda$ & $\rho_s$ & $\delta$ & Wald$_t$ & Proposed$_{\chi^2}$ & NC & B & I & L \\
\midrule
 10 & 0.4 & $1/2$ & 0 & 0 & 0.145 & 0.171 & 1/2 & 78 & 0 & 0 \\
 10 & 0.4 & $1/2$ & 0.4 & 0 & 0.140 & 0.180 & 0/0 & 59 & 0 & 0 \\
 10 & 0.4 & $1/2$ & 0 & 0.5 & 0.422 & 0.524 & 0/4 & 75 & 0 & 0 \\
 10 & 0.4 & $1$ & 0 & 0 & 0.111 & 0.153 & 1/0 & 48 & 0 & 0 \\
 10 & 0.4 & $1$ & 0.4 & 0 & 0.128 & 0.180 & 0/2 & 37 & 0 & 0 \\
 10 & 0.4 & $1$ & 0 & 0.5 & 0.466 & 0.574 & 0/4 & 50 & 0 & 0 \\
 10 & 0.4 & $2$ & 0 & 0 & 0.101 & 0.171 & 0/1 & 65 & 0 & 0 \\
 10 & 0.4 & $2$ & 0.4 & 0 & 0.097 & 0.147 & 0/1 & 48 & 0 & 0 \\
 10 & 0.4 & $2$ & 0 & 0.5 & 0.405 & 0.529 & 1/1 & 68 & 0 & 0 \\
\bottomrule
\end{tabular}
\end{table}

\begin{table}[p]
\centering
\caption{Auxiliary statistics and computational diagnostics of the additional settings with $k=20$ and $\rho=0.4$ (1000 replicates per setting). Columns as in Table~\ref{tab:sim_checks_b}.}
\label{tab:sim_checks_c}
\scriptsize
\setlength{\tabcolsep}{4pt}
\begin{tabular}{@{}ccccc cc cccc@{}}
\toprule
$k$ & $\rho$ & $\lambda$ & $\rho_s$ & $\delta$ & Wald$_t$ & Proposed$_{\chi^2}$ & NC & B & I & L \\
\midrule
 20 & 0.4 & $1/2$ & 0 & 0 & 0.125 & 0.138 & 1/2 & 2 & 0 & 0 \\
 20 & 0.4 & $1/2$ & 0.4 & 0 & 0.116 & 0.128 & 0/0 & 0 & 0 & 0 \\
 20 & 0.4 & $1/2$ & 0 & 0.5 & 0.696 & 0.738 & 1/1 & 0 & 0 & 0 \\
 20 & 0.4 & $1$ & 0 & 0 & 0.116 & 0.129 & 0/3 & 0 & 0 & 0 \\
 20 & 0.4 & $1$ & 0.4 & 0 & 0.107 & 0.120 & 0/1 & 0 & 0 & 0 \\
 20 & 0.4 & $1$ & 0 & 0.5 & 0.749 & 0.780 & 0/1 & 0 & 0 & 0 \\
 20 & 0.4 & $2$ & 0 & 0 & 0.107 & 0.132 & 1/0 & 1 & 0 & 0 \\
 20 & 0.4 & $2$ & 0.4 & 0 & 0.118 & 0.145 & 0/0 & 2 & 0 & 0 \\
 20 & 0.4 & $2$ & 0 & 0.5 & 0.684 & 0.726 & 0/2 & 2 & 0 & 0 \\
\bottomrule
\end{tabular}
\end{table}

\begin{table}[p]
\centering
\caption{Auxiliary statistics and computational diagnostics of the additional settings at $k=30$ with $\rho=0$ or $0.8$, and at $k=50$ with $\rho=0.4$ (1000 replicates per setting). Columns as in Table~\ref{tab:sim_checks_b}.}
\label{tab:sim_checks_d}
\scriptsize
\setlength{\tabcolsep}{4pt}
\begin{tabular}{@{}ccccc cc cccc@{}}
\toprule
$k$ & $\rho$ & $\lambda$ & $\rho_s$ & $\delta$ & Wald$_t$ & Proposed$_{\chi^2}$ & NC & B & I & L \\
\midrule
 30 & 0 & $1/2$ & 0 & 0 & 0.110 & 0.124 & 0/1 & 0 & 0 & 0 \\
 30 & 0 & $1/2$ & 0.4 & 0 & 0.106 & 0.110 & 0/1 & 0 & 0 & 0 \\
 30 & 0 & $1/2$ & 0 & 0.5 & 0.713 & 0.738 & 0/1 & 0 & 0 & 0 \\
 30 & 0 & $1$ & 0 & 0 & 0.110 & 0.126 & 1/3 & 0 & 0 & 0 \\
 30 & 0 & $1$ & 0.4 & 0 & 0.126 & 0.136 & 0/0 & 0 & 0 & 0 \\
 30 & 0 & $1$ & 0 & 0.5 & 0.721 & 0.746 & 0/0 & 0 & 0 & 0 \\
 30 & 0 & $2$ & 0 & 0 & 0.085 & 0.096 & 1/4 & 0 & 0 & 0 \\
 30 & 0 & $2$ & 0.4 & 0 & 0.108 & 0.114 & 0/1 & 0 & 0 & 0 \\
 30 & 0 & $2$ & 0 & 0.5 & 0.699 & 0.731 & 0/0 & 0 & 0 & 0 \\
\addlinespace[2pt]
 30 & 0.8 & $1/2$ & 0 & 0 & 0.108 & 0.106 & 0/2 & 7 & 0 & 0 \\
 30 & 0.8 & $1/2$ & 0.4 & 0 & 0.099 & 0.109 & 1/3 & 6 & 0 & 0 \\
 30 & 0.8 & $1/2$ & 0 & 0.5 & 0.963 & 0.967 & 0/0 & 3 & 0 & 0 \\
 30 & 0.8 & $1$ & 0 & 0 & 0.110 & 0.124 & 0/1 & 0 & 0 & 0 \\
 30 & 0.8 & $1$ & 0.4 & 0 & 0.116 & 0.122 & 0/0 & 0 & 0 & 0 \\
 30 & 0.8 & $1$ & 0 & 0.5 & 0.982 & 0.986 & 0/0 & 0 & 0 & 0 \\
 30 & 0.8 & $2$ & 0 & 0 & 0.110 & 0.118 & 1/0 & 6 & 0 & 0 \\
 30 & 0.8 & $2$ & 0.4 & 0 & 0.123 & 0.138 & 0/1 & 4 & 0 & 0 \\
 30 & 0.8 & $2$ & 0 & 0.5 & 0.982 & 0.984 & 0/2 & 3 & 0 & 0 \\
\addlinespace[4pt]
 50 & 0.4 & $1/2$ & 0 & 0 & 0.109 & 0.115 & 0/0 & 0 & 0 & 0 \\
 50 & 0.4 & $1/2$ & 0.4 & 0 & 0.104 & 0.113 & 0/0 & 0 & 0 & 0 \\
 50 & 0.4 & $1/2$ & 0 & 0.5 & 0.965 & 0.965 & 1/1 & 0 & 0 & 0 \\
 50 & 0.4 & $1$ & 0 & 0 & 0.098 & 0.107 & 0/0 & 0 & 0 & 0 \\
 50 & 0.4 & $1$ & 0.4 & 0 & 0.086 & 0.092 & 0/0 & 0 & 0 & 0 \\
 50 & 0.4 & $1$ & 0 & 0.5 & 0.973 & 0.974 & 0/0 & 0 & 0 & 0 \\
 50 & 0.4 & $2$ & 0 & 0 & 0.121 & 0.130 & 0/2 & 0 & 0 & 0 \\
 50 & 0.4 & $2$ & 0.4 & 0 & 0.095 & 0.098 & 0/0 & 0 & 0 & 0 \\
 50 & 0.4 & $2$ & 0 & 0.5 & 0.965 & 0.968 & 0/2 & 0 & 0 & 0 \\
\bottomrule
\end{tabular}
\end{table}

\section{Data and sensitivity analyses}\label{sapp:data}

\subsection{The FIT review}\label{sapp:fit}

The Cochrane review of Grobbee et al.\cite{grobbee2022} compares guaiac-based faecal occult blood tests with faecal immunochemical tests (FIT) for colorectal cancer in average-risk screening populations and separates the studies in which every participant underwent the reference standard (``reference standard: all'') from those in which only test-positive participants underwent colonoscopy while test-negative participants were followed for at least one year for interval cancers (``reference standard: positive''). The main text uses the review's FIT analysis of the second group for the target condition colorectal cancer, whose $2\times2$ tables are the published counts of the review's own analysis (its RevMan data file, analysis ``Reference standard: positive, FIT, CRC''). The review contains many analyses (test, reference-standard design, target condition and subgroups); this one was chosen because it has the largest studies, the widest range of sizes and the fewest empty cells of the review's FIT analyses, and the review itself was identified in a screen of Cochrane DTA reviews with published data files for analyses with at least 20 entries and an asymmetric fitted curve. Table~\ref{tab:fit_counts} lists the entries with the country, the positivity threshold and the follow-up of test negatives as the review's characteristics tables report them. Twenty-two studies contribute 23 entries: Chiang 2014 evaluated two FIT brands in different parts of the Taiwanese programme (747,076 and 208,929 participants). These are the only two entries from the same study and were analysed as separate units, as in the source review. Denters 2012a and Levi 2011b are the FIT arms of studies that randomized participants between a guaiac test and a FIT. The observed false-positive rates of the entries range from $2.0\%$ to $11.5\%$ (Table~\ref{tab:fit_counts}), the range of false-positive rates over which the fitted curve is supported by data. Four entries have no false negatives. In these screening populations the prevalence of cancer is below one per cent, so $\ESS_i=4n_{1i}n_{0i}/(n_{1i}+n_{0i})$ is close to four times the number of cancers and $s_i=1/\sqrt{\ESS_i}$ orders the entries by that number: the seven smallest entries have 6 to 16 cancers and the seven largest 385 to 1,578. The false-positive rate averaged 7.3\% among the seven smallest entries and 4.4\% among the seven largest, with mean sensitivities of 87\% and 85\%, respectively. The Cochrane Library's terms for downloaded data grant a non-commercial licence to extract, copy and share the data with attribution.

\begin{table}[p]
\centering
\caption{The 23 entries of the FIT analysis (reference standard: positive; target condition colorectal cancer) of Grobbee et al.\cite{grobbee2022}, with the country, positivity threshold and follow-up of test-negative participants transcribed from the review's characteristics tables, and the counts of the review's analysis (TP, FN among cancers; FP, TN among participants without cancer). Thresholds in $\mu$g Hb/g faeces where the review states or converts them, with the manufacturer's buffer concentration in parentheses where the review gives it (ng/mL, or $\mu$g/L, equivalently); qualitative tests have no quantitative threshold. Follow-up of test-negative participants in years. Study labels are the review's.}
\label{tab:fit_counts}
\scriptsize
\setlength{\tabcolsep}{3pt}
\begin{tabular}{llllrrrr}
\toprule
Entry & Country & Threshold ($\mu$g Hb/g) & Follow-up & TP & FN & FP & TN \\
\midrule
Arana-Arri 2017a & Spain & 20 & 2 y & 1032 & 136 & 14831 & 279546 \\
Burón 2019 & Spain & 20 & 2 y & 251 & 47 & 3931 & 74901 \\
Castiglione 2007 & Italy & 20 (100 ng/mL) & 1--2 y & 65 & 16 & 894 & 26390 \\
Chen 2016a & Taiwan & 20 (100 ng/mL) & 2 y & 712 & 51 & 7871 & 132411 \\
Chiang 2014a & Taiwan & 20 & $>$3 y & 1197 & 381 & 21539 & 723959 \\
Chiang 2014b & Taiwan & 20 & $>$3 y & 284 & 101 & 6639 & 201905 \\
Crotta 2012 & Italy & 20 (100 ng/mL) & 2 y & 5 & 3 & 62 & 1585 \\
Denters 2012a & Netherlands & not stated & 2 y & 12 & 4 & 174 & 2634 \\
Digby 2016 & Scotland & 80 & 2 y & 30 & 31 & 619 & 30109 \\
Itoh 1996 & Japan & 10 (50 ng/mL) & 2 y & 77 & 12 & 1130 & 26358 \\
Jensen 2016 & USA & 20 & 1 y & 514 & 100 & 11598 & 307212 \\
Juul 2018 & Denmark & 20 (100 $\mu$g/L) & 1 y & 907 & 60 & 13955 & 229093 \\
Kapidzic 2017 & Netherlands & 10 & $\ge$2 y & 13 & 0 & 213 & 1636 \\
Launoy 2005 & France & 67 (20 ng/mL) & $\ge$1 y & 22 & 4 & 344 & 6985 \\
Levi 2011b & Israel & 14 (70 ng/mL) & 2 y & 6 & 0 & 102 & 1071 \\
McNamara 2014 & Ireland & 20 (100 ng/mL) & 2 y & 16 & 1 & 402 & 4549 \\
Nakama 1996 & Japan & not stated (qualitative) & 3 y & 10 & 4 & 147 & 3204 \\
Parente 2013 & Italy & 250 (100 ng/mL) & $\ge$2 y & 95 & 8 & 1957 & 36293 \\
Parra-Blanco 2010 & Spain & 50 ng/mL (qualitative) & 2 y & 14 & 0 & 129 & 1573 \\
Robinson 1996 & England & 1:8 dilution (qualitative) & not stated & 10 & 0 & 135 & 1344 \\
Sieg 2002 & Germany & 10; 5 from 1998 & 1--4 y & 23 & 6 & 206 & 5053 \\
Van Roon 2013 & Netherlands & 10 (50 ng/mL) & 2 y & 22 & 2 & 338 & 4141 \\
Zorzi 2018 & Italy & 20 (100 ng/mL) & 2 y & 412 & 51 & 6007 & 116281 \\
\bottomrule
\end{tabular}
\end{table}

\emph{Leave-one-out fits.} Omitting one entry at a time and refitting gives a shape $\hat\lambda$ between $2.00$ and $2.20$, a latent accuracy trend between $-4.2$ and $-2.9$ with likelihood-ratio $p$ between $0.015$ and $0.093$ (below $0.10$ in all 23 fits), a latent threshold trend between $+3.0$ and $+5.4$, and a Deeks $p$ between $0.106$ and $0.562$ (above $0.10$ in all 23 fits); the binomial-fit $\lnDOR$ test has $p$ between $0.488$ and $0.955$. The largest single change comes from omitting Denters 2012a ($p=0.093$ for the accuracy test); omitting the largest entry (Chiang 2014a, 1,578 cancers) leaves the accuracy test at $p=0.060$ and the Deeks test at $p=0.106$. Chiang 2014a and 2014b evaluated OC-Sensor and HM-Jack, respectively, in separate participant groups. The two leave-one-out fits that omitted either entry therefore included one entry per study. In both fits, the accuracy test rejected at the 10\% level ($p=0.060$ and $0.027$), whereas Deeks did not ($p=0.106$ and $0.430$). The high fitted residual correlation ($\hat\rho=+0.98$) corresponds to little residual variation in latent accuracy after accounting for the study-size trends. The estimated residual standard deviations are $\hat\sigma_\alpha=0.095$ for latent accuracy and $\hat\sigma_\theta=0.52$ for latent threshold ($\sigma_\alpha/\sigma_\theta=\{4(1-\rho)/(1+\rho)\}^{1/2}$); no parameter of the fit is at a bound of the optimizer (Table~\ref{tab:review_fits}).

\emph{Parametric check in the design of the review.} A parametric check evaluated the proposed test at the accuracy-null constrained fit of Section~\ref{sapp:lrt}. Using the observed group sizes of the 23 entries, $2{,}000$ replicate data sets were generated with $\gamma_\alpha=0$ and the remaining parameters at their constrained estimates ($\hat\gamma_\theta=5.5$, $\hat\lambda=2.14$ and $\hat\rho=0.979$). Every replicate was analysed with the same fitting and testing procedure as the review, re-estimating the full and accuracy-constrained models and leaving the threshold trend free in both (\texttt{analysis/fit\_level\_checks.py} in the accompanying reproducibility archive). At the nominal level $0.10$, the rejection rate was $0.095$ (Monte-Carlo SE $0.007$). Of the replicates, $5$ had non-convergence flags and $391$ had correlation estimates at the optimizer's bound ($\rho=0.9993$); these replicates were retained. This is a local numerical check of the implemented procedure at the specified generating point.

\subsection{The IPG review: reconstruction of the tables}\label{sapp:ipg}

The IPG source\cite{goodacre2006} reports each cohort's sensitivity and specificity as proportions rounded to two decimals with exact (Clopper--Pearson) 95\% intervals, but not the numbers of diseased and non-diseased participants. The two Prandoni 1991 cohorts are the derivation and validation cohorts from the same publication and are analysed separately, as in the source review.

We recover integer $2\times2$ tables by inversion: for every printed triple (proportion, lower limit, upper limit) we enumerate the integer pairs (count, group size) up to a group size of $4000$ whose proportion is reproduced within half-unit tolerance ($|x/n-p_{\mathrm{reported}}|\le0.005+10^{-9}$, robust to half-up versus banker's rounding) and whose exact interval matches the printed endpoints to the printed precision. For each sensitivity and specificity coordinate, we retain up to twelve admissible (count, group size) pairs, ordered by increasing group size and then count; this limit was binding for $4$ sensitivity and $11$ specificity coordinates. We then combine the retained sensitivity and specificity candidates, requiring an implied prevalence between $0.02$ and $0.95$. Within these restricted candidate sets, $30$ of the $42$ cohorts admit more than one table.

For each cohort, we select the combinations with the smallest and largest total sample sizes. We carry these two reconstructions through the \emph{entire} pipeline: model fit, the accuracy coordinate, all tests. The main text uses the minimum-count reconstruction throughout. Table~\ref{tab:review_fits} reports results for the two reconstructions examined. Larger group sizes than those retained would give the affected cohorts more weight than the minimum-count tables do.

Table~\ref{tab:review_fits} gives the fitted parameters of both reviews, at both IPG endpoints: the trend vector and the residual correlation, which Table~2 of the main text omits, the shape, and the tests, which it reports for the minimum-count endpoint. The estimates were similar at the two endpoints. The two IPG endpoints give $\hat\lambda=1.00$ and $0.99$, a latent accuracy trend of $-9.7$ and $-9.4$ ($p=0.023$ and $0.023$), a latent threshold trend of $+6.3$ and $+6.1$ ($p=0.013$ and $0.014$), a $\lnDOR$ trend of the fit with $p=0.010$ and $0.011$, and Deeks $p=0.012$ and $0.016$. All three tests gave the same conclusions at the two reconstruction endpoints. Full per-study reconstruction listings, both endpoints, and the reconstruction scripts are included in the accompanying reproducibility archive.

\begin{table}[t]
\centering
\caption{Fitted parameters and test results for FIT and for the minimum-count and maximum-count IPG reconstructions within the restricted candidate sets of Section~\ref{sapp:ipg}. Trends are per unit of $s_i=1/\sqrt{\ESS_i}$. Standard errors and test references follow Section~2.4 of the main text. Deeks results come from the empirical-logit weighted regression; model parameters and the other tests come from the binomial fit. The table also reports latent residual standard deviations and whether an optimizer bound was reached.}
\label{tab:review_fits}
\scriptsize
\setlength{\tabcolsep}{3pt}
\begin{tabular}{lccc}
\toprule
 & FIT (colorectal cancer) & \multicolumn{2}{c}{IPG (DVT)} \\
\cmidrule(lr){2-2}\cmidrule(lr){3-4}
 & published counts & min & max \\
\midrule
Entries $k$; entries with an empty cell & 23;\ 4 & 42;\ 2 & 42;\ 2 \\
\addlinespace
Trend vector $\hat\beta_\eta$ (SE) & $+2.9$\ (3.5) & $+1.5$\ (2.8) & $+1.4$\ (2.7) \\
\quad $\hat\beta_\varphi$ (SE) & $+3.9$\ (1.3) & $+11.1$\ (2.9) & $+10.8$\ (2.9) \\
Intercepts $\hat\mu_\eta$; $\hat\mu_\varphi$ & 1.91; -2.93 & 1.28; -1.91 & 1.28; -1.91 \\
Residual SDs $\hat\sigma_\eta$; $\hat\sigma_\varphi$ & 0.76; 0.36 & 0.69; 0.69 & 0.70; 0.70 \\
Residual correlation $\hat\rho$ & $+0.98$ & $+0.40$ & $+0.39$ \\
Latent SDs $\hat\sigma_\alpha$; $\hat\sigma_\theta$ & 0.095; 0.52 & 0.76; 0.58 & 0.77; 0.58 \\
Any parameter at a bound of the optimizer & no & no & no \\
Shape $\hat\bH$ (SE) & $-0.75$\ (0.09) & $+0.00$\ (0.21) & $+0.01$\ (0.21) \\
$\hat\lambda$ [95\% CI] & 2.11\ [1.73, 2.56] & 1.00\ [0.65, 1.53] & 0.99\ [0.65, 1.52] \\
\addlinespace
Proposed: latent accuracy trend $\hat\gamma_\alpha$ (SE); $p$ & $-3.7$\ (1.5); 0.029 & $-9.7$\ (3.9); 0.023 & $-9.4$\ (3.8); 0.023 \\
Latent threshold trend $\hat\gamma_\theta$ (SE); $p$ & $+3.9$\ (2.1); 0.081 & $+6.3$\ (2.3); 0.013 & $+6.1$\ (2.3); 0.014 \\
\addlinespace
$\lnDOR$ trend of the fit (SE); $p$ & $-1.0$\ (2.7); 0.710 & $-9.7$\ (3.6); 0.010 & $-9.4$\ (3.5); 0.011 \\
Deeks test: slope (SE); $p$ & $-3.5$\ (4.1); 0.406 & $-10.7$\ (4.1); 0.012 & $-10.2$\ (4.0); 0.016 \\
\bottomrule
\end{tabular}
\end{table}

\end{document}